\documentclass[10pt,twocolumn,twoside]{IEEEtran}

\usepackage[cmex10]{amsmath}
\usepackage{amssymb,amsthm}
\usepackage{cite}
\usepackage{graphicx}
\usepackage{epsfig}
\usepackage{url}
\usepackage{color}
\usepackage{setspace,epstopdf}
\usepackage{algorithm}
\usepackage{algpseudocode}
\usepackage{verbatim}
\usepackage{multirow}
\usepackage{subcaption}
\usepackage{thmtools,thm-restate}

\newcommand{\ignore}[1]{}

\newtheorem{theorem}{Theorem}[section]
\newtheorem{proposition}{Proposition}[section]
\newtheorem{corollary}{Corollary}[section]
\newtheorem{lemma}{Lemma}[section]
\newtheorem{definition}{Definition}[section]

\renewcommand{\baselinestretch}{0.93}
\newif\ifreport\reportfalse
\newcommand\numberthis{\addtocounter{equation}{1}\tag{\theequation}}
\begin{document}

\title{Defending Against Stealthy Attacks on Multiple Nodes with Limited Resources: A Game-Theoretic Analysis\thanks{This work has been funded by QNRF fund NPRP 5-559-2-227, ARO-W911NF-15-1-0277, NSF grant CNS-1816495, and a grant from the Board of Regents of the State of Louisiana LEQSF(2017-19)-RD-A-15.}}

\author{Ming~Zhang,~\IEEEmembership{Student Member,~IEEE,}
        Zizhan~Zheng,~\IEEEmembership{Member,~IEEE,}
        and~Ness~B.~Shroff,~\IEEEmembership{Life~Fellow,~IEEE}
\thanks{M. Zhang (zhang.2562@osu.edu) is with the Department of Computer Science and Engineering, Ohio State University, Columbus,
OH, 43202 USA.}
\thanks{Z. Zheng (zzheng3@tulane.edu) is with the Department of Computer Science, Tulane University, New Orleans, LA, 70118, USA.}
\thanks{N. Shroff (shroff.11@osu.edu) is with the Departments of ECE and CSE, Ohio State University, Columbus, OH, 43202, USA.}}

\maketitle

\begin{abstract}
Stealthy attacks have become a major threat for cyber security. Previous works in this direction fail to capture the practical resource constraints and mainly focus on one-node settings. In this paper, we propose a two-player game-theoretic model including a system of multiple independent nodes, a stealthy attacker and an observable defender. In our model, the attacker can fully observe the defender's behavior and the system state, while the defender has zero feedback information. Further, a strict resource constraint is introduced to limit the frequency of the attacks/defenses for both players. We characterize the best responses for both attacker and defender under both non-adaptive and adaptive strategies. We then study the sequential game where the defender first announces its strategy and the attacker then responds accordingly. We have designed an algorithm that finds a nearly optimal strategy for the defender and provided a full analysis of its complexity and performance guarantee.
\end{abstract}


\section{Introduction}

Increasingly sophisticated cyber attacks constantly push the evolution of cyber security. In recent years, worldwide organizations and IT companies, e.g., United Nation, Google and Amazon, are facing a significantly increasing number of {\it Advanced Persistent Threats} (APT)~\cite{APT}. The APT attack has several distinguishing properties that render traditional defense mechanisms less effective. First, they are often launched by {\it incentive driven} entities, including government and competitive companies with specific targets. Second, the APT attack is {\it persistent}, which usually involves multiple stages and frequent compromises of the system. Based on \cite{fireeyereport2018}, half of the entities suffering APT attacks experienced another successful compromise within one year. Third, they are highly adaptive and {\it stealthy}, often operating in a ``low-and-slow'' fashion~\cite{graph-apt} in order to maintain a small footprint and avoid of being detected. In fact, some of the past APT attacks have been so effective because they have gone undetected for months or longer~\cite{CDorked,Gauss}. Hence, conventional security measures against one-shot attack and known attack types are not sufficient in the face of long-lasting and stealthy attacks. Meanwhile, the objective of APT attacks usually includes the key information theft and complete control over the system, resulting in a much bigger loss than traditional cyber attacks.


In this paper, we study a two-player non-zero-sum game that explicitly models stealthy attacks with resource constraints, as an extension of the asymmetric version of the FlipIt game considered in~\cite{asymmetric-model} . We consider a system with $N$ {\it independent} nodes (or components), an attacker, and a defender. Both players compete for the control of the system by attacking or defending each node, subject to an instantaneous move cost per node and a long-term average resource constraint across the entire system. The attacker tries to maximize its benefits by successfully compromising nodes, and the defender aims at minimizing the total defense cost and value loss incurred by losing control of a node.

To model the stealthy attacks, we assume that the defender has no feedback about the node state and the attacker's behavior across the entire game, which is reasonable in many security setups. On the other hand, the attacker is capable of observing the defender's each move as well as the node state, and makes decisions accordingly. In this work, we consider two commonly adopted solution concepts, Nash Equilibrium and Sequential Equilibrium, both of which have been applied to cybersecurity. In the former, the defender and the attacker determine their strategies at the beginning of the game simultaneously, while in the latter, the defender acts as the leader of the game and commits to a strategy first, and the attacker as the follower then responds accordingly.

For tractability and simplifying the analysis, we assume that the set of nodes are 
{\it independent} in the sense that the proper functioning of one node does not depend on other nodes, which serves as a first-order approximation of the more general setting of interdependent nodes to be considered in our future work. Despite of the assumption that each node is independent, the multi-node setting together with the resource constraints impose significant challenges in characterizing the best responses, Nash Equilibria and Sequential Equilibria of the games.



One example where our game model can be applied is key rotation. For a system with multiple communication links or servers that are protected by different keys, an APT attacker may compromise some of the keys from time to time. A common practice is to periodically generate fresh keys by a trusted key-management service, without knowing when they are compromised. On the other hand, the attacker can easily detect when the key expires with a negligible cost and there is a constraint on the frequency of moves at both sides. There are also other examples where our model can be useful such as password reset and virtual machine refreshing~\cite{flipit,flipit-2,asymmetric-model}.

To help reader better understand our main results, we briefly explain the key concepts below. Formal definitions can be found in Sections~\ref{sec:model} and \ref{sec:best-response}.
\begin{itemize}
\item In a {\it periodic defense} strategy: the defender protects each node periodically. That is, the time interval between two consecutive defenses is fixed for a given node.
\item In an {\it i.i.d. attack} strategy : the attacker's waiting time before each attack (modeled as a random variable) is {\it i.i.d.} across time.
\item In a {\it Markovian defense (resp. attack)} strategy: the time interval between two defenses (resp. the waiting time of each attack) follows a Markov process.
\item Nearly Optimal strategy: For arbitrary small positive number $\epsilon$, we can always find a strategy that the performance difference between this strategy and the theoretical optimal strategy is less than $\epsilon$.
\end{itemize}

We have made following contributions in this paper with the main results summarized in Table~\ref{tbl:mainresults}~\footnote{In Table~\ref{tbl:mainresults}, $A \rightarrow B$ means that $B$ is a best response against $A$; $A \nrightarrow B$ means that $B$ is NOT a best response against $A$}.
\begin{itemize}
\item We propose a two-player game model with multiple independent nodes, an overt defender, and a stealthy attacker where both players have strict resource constraints. 
\item We prove that periodic defense is a best response against {\it i.i.d.} attack among {\it all} defense strategies, and {\it i.i.d.} attack is a best response against periodic defense among {\it all} attack strategies. We further consider Markovian strategies and prove that periodic defense is still a best response against a Markovian attacking strategy, but {\it i.i.d.} attack is not necessarily a best response against a Markovian defending strategy.
\item For the pair of periodic defense and {\it i.i.d.} attack strategies, we fully characterize the set of Nash Equilibria of our game, and show that there is always one (and maybe more) equilibrium, for the case when the attack times are deterministic.
\item We further consider the sequential game with the defender as the leader and the attacker as the follower. We design a dynamic programming based algorithm 
that identifies a nearly optimal strategy (in the sense of subgame perfect equilibrium) for the defender. We also fully characterize the trade-off between algorithm performance and its complexity.
\end{itemize}

\begin{table}[!t]
\centering
\caption{Main Results}
\renewcommand{\arraystretch}{1.05}
\label{tbl:mainresults}
\small{
\begin{tabular}{|c|c|c|}
\hline \multicolumn{1}{|c|}{} & \multicolumn{1}{|c|}{Attacker} & \multicolumn{1}{|c|}{Defender}\\
\hline
\multirow{3}{*}{Best Response} & \multicolumn{2}{|c|}{{\it i.i.d.} attack $\leftrightarrows$ periodic defense} \\ \cline{2-3}
                               & \multicolumn{2}{|c|}{Markovian attack $\rightarrow$ periodic defense} \\ \cline{2-3}
                               & \multicolumn{2}{|c|}{{\it i.i.d.} attack $\nleftarrow$ Markovian defense} \\ \cline{2-3}
\hline
Nash Equilibrium & \multicolumn{2}{|c|}{A complete characterization of NEs (6 types)}\\
\hline
 \multirow{3}{*}{Sequential Game} & Optimal attack		  		   & A polynomial time  \\
 								  & under a given                  & algorithm for optimal \\
 								  & defense strategy~\eqref{snegame2}      & defense (Algorithm~\ref{alg:snefinal}) \\
\hline
\end{tabular}
}
\end{table}

This paper is the extended version of \cite{zhang2015game}. In addition to improving the presentation and organization of the paper, we have provided in this journal submission version (i) an extension of the defender's and attacker's best response strategies from the non-adaptive setting to the general adaptive setting, (ii) we provide some preliminary analysis about Markovian strategies for both attacker and defender, and (iii) a better understanding of the performance vs. complexity trade-off of our algorithm for the sequential game, reducing its complexity by a factor of $O(N^3)$ with the same performance guarantee. 
\ifreport
\begin{itemize}
\item We have extended the strategy spaces for both defender and attacker to adaptive strategies. The optimality of deterministic defense strategies against any attacking strategy is established and the independent non-adaptive attacking strategies are proven optimal against any deterministic defense strategy.
\item We have studied the best response strategies for both attacker and defender when the other player employs a Markovian strategy. We have shown that the periodic defending strategy is still optimal against any Markovian attacking strategy.
    However, the {\it i.i.d.} attacking strategy is not optimal against a Markovian defending strategy.
\item For the sequential game, we have reduced the complexity of Algorithm~\ref{alg:snefinal} by a factor of $O(N^3)$ with a detailed discussion on the performance vs. complexity trade-off.
\end{itemize}
\fi

The remainder of this paper is organized as follows. A summary of related work is provided in Section~\ref{sec:related}.
We present our game-theoretic model in Section~\ref{sec:model}, and study best-response strategies of both players in Section~\ref{sec:best-response}. 
The sequential game is studied in Section~\ref{sec:sequential}. In Section~\ref{sec:numerical}, we present numerical result, and we conclude the paper in Section~\ref{sec:conclusion}.

\section{Related Work}\label{sec:related}
Game theory has been extensively applied to cyber-security and network security~\cite{Basar-network-security-book,Basar-network-security-survey,
interdependent-security,
Tambe-security-game}. However, traditional models mainly focus on known attacks and largely ignore the budget constraints of both the defender and the attacker.

As mentioned in the introduction, our model is inspired by the FlipIt game~\cite{flipit,flipit-2} proposed in response to an APT attack towards RSA Data Security~\cite{RSA}, a non-zero-sum dynamic game that explicitly models the stealthy takeover 
of a single node. 
In the original model, a player obtains control over a component instantaneously by ``flipping" it, and obtains feedback only when it moves. Dominant strategies or strongly dominant strategies are characterized for several classes of periodic and renewable strategies and some simple 
adaptive strategies. But the full analysis of Nash Equilibrium is only provided when both the defender and the attacker employ a periodic strategy with a random starting phase. Several variants of the basic model have been studied~\cite{FlipThem,asymmetric-model}. In particular, a multi-node extension is considered in~\cite{FlipThem} where the attacker needs to compromise either all the nodes (AND model) or a single node (OR model) to take over a system. The authors name such a model as ``FlipThem''. However, only preliminary analytic results are provided. Leslie et al. extend the ``FlipThem'' model in \cite{leslie2015threshold,leslie2017multi} where the attacker can obtain partial benefits by compromising a certain number (larger than a threshold) of nodes. 
An asymmetric model similar to ours where the attacker is stealthy while the defender is observable is considered in~\cite{asymmetric-model}, where full Nash Equilibrium analysis is provided but only for the single node setting. In \cite{nochenson2013behavioral}, Nochenson et al first initiate the effort of adding player's characterization information including gender and age in the FlipIt game model. Basak et al \cite{basak2018initial} further extend the concept by adding different type of rationale of human agents. In \cite{zheng2017reset}, Zheng et al. use multi-armed bandit model to investigate the optimal timing of security updates against stealthy attacks. There are also some behaviorial studies of the FlipIt game~\cite{FlipIt-behavior}. However, none of the previous works considered an explicit resource constraint on the players.

A different type of security game has also been studied in the literature mainly for protecting physical infrastructures~\cite{Tambe-AAMAS08,Tambe-security-game,Tambe-JAIR11,Tambe-AAMAS13}. Essentially a mixed strategy Stackelberg game is considered, where the defender is the leader and the attacker is the follower. The key assumption is that the defender first decides upon a randomized defense policy, and the attacker then observes the randomized policy of the defender but not its realization before taking an action. While this is a useful assumption under certain scenarios, it may not hold when the attacker is highly adaptive. In particular, since the attacker may be able to observe the defender's previous actions, it could take an action before the defender changes its policy to get more benefit. Moreover, the two-stage game is insufficient to capture the persistent and stealthy behaviors of advanced attacks. In spite of the fundamental differences of the two models, recent work that extend this model to multiple defenders and bounded rationality~\cite{Tambe-IJCAI13,Tambe-AAAI13} provide useful insights to our model as well, which will be studied in our future work.

\section{Game Model}\label{sec:model}
In this section, we discuss our two-player game model including its information structure, the action spaces of both attacker and defender, and their payoffs. Our game model extends the single node model in~\cite{asymmetric-model} to multiple nodes and includes a resource constraint on each player.

\subsection{Basic Model}\label{sec:basicmodel}
In our game-theoretic model, there are two players (the defender and the attacker) and a network of $N$ {\it independent} nodes\footnote{The terms ``components'' and ``nodes'' are interchangeable in this paper.}. Each node has a value of $r_i$ representing the payoff the attacker can receive per unit time by successfully compromising node $i$.  
We consider finite time horizon where the game starts at time $t=0$ and goes to any time $t=T$. We assume that time is continuous. 
Every time when the attacker starts an attack for node $i$, it incurs a cost of $C^A_i$ and takes a random period of time $\alpha_{i,k}$ to succeed. On the other hand, if the defender makes a move to protect node $i$, the node is immediately recovered and incurring a cost of $C_i^D$. Further, this information is immediately learned by the attacker. The attacker's strategy is to determine 
$W_{i,k}$, 
the waiting time from the defender's $k$-th move to its next attack on node $i$, for each $i$ and $k$. On the contrary, the defender's strategy is to determine the time intervals between its $(k-1)$-th move and $k$-th move for each node $i$ and $k$, denoted as $X_{i,k}$. 
Both the attacker's and the defender's strategies can be randomized and adaptive in general.

In this paper, an attack strategy is considered {\it adaptive} when the attacker's decision on $W_{i,k}$ for any $i$ and $k$ can depend on the realized value of $X_{i',k'}$ for any $i'$ and $k' \leq k$. An adaptive defense strategy is defined similarly. On the other hand, a strategy is defined as {\it non-adaptive} if the values of $W_{i,k}$'s and $X_{i,k}$'s are either pre-computed or follow fixed probability distributions. 
The attacker (defender) can attack (defend) multiple nodes at the same time and maintain their possession until the other player's next move, which may or may not change the node state.


In addition to the move cost, we introduce a strict resource constraint for each player, which is a practical assumption but has been ignored in most prior works on security games. In particular, we place an upper bound on the average amount of resource that is available to each player at any time (to be formally defined below). As in typical security games, we assume that $r_i, C^A_i, C^D_i$, the distribution of $\alpha_{i,k}$, and the budget constraints are all common knowledge of the game, that is, they are known to both players. 
Without loss of generality, all nodes are assumed to be protected at time $t=0$. Table~\ref{notation} summarizes the notations used in the paper.

\begin{figure}[!t]
\centering
\includegraphics[width=0.8\linewidth]{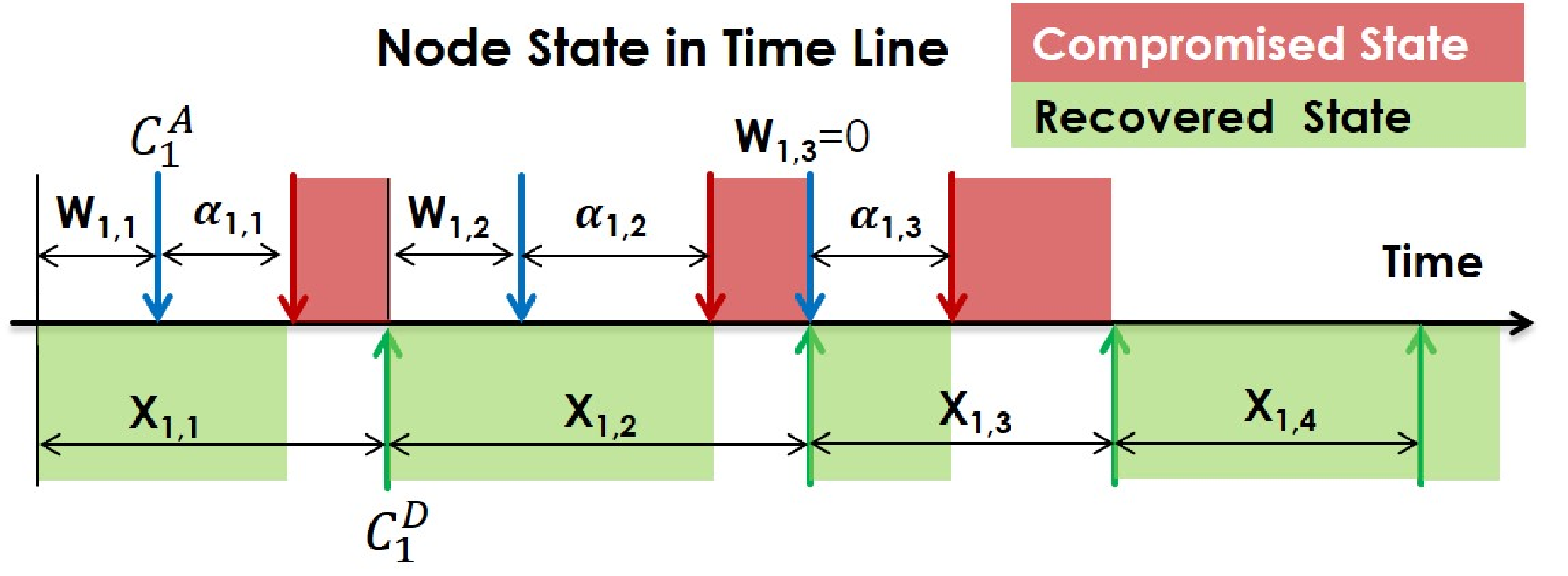}
\caption{Game Model}
\label{model}
\end{figure}

As in~\cite{asymmetric-model}, we consider an asymmetric feedback model where the attacker's moves are {\it stealthy}, while the defenders' moves are {\it observable}. More specifically, at any time, the attacker knows the full history of moves by the defender, as well as the state of each node, while the defender does not know whether a node is compromised or not. This asymmetric information structure is crucial in modeling stealthy attacks in cyber security. 

In this paper, we consider both non-adaptive and adaptive strategies. We define the strategy space for the attacker as all possible $W_{i,k}\ \forall i,k$ that follows a joint distribution. Similarly, the defender's strategy space refers to all possible $X_{i,k}\ \forall i,k$ following a joint distribution. Since the defender cannot observe the attacker's behavior and node states, we only need to consider non-adaptive strategies for the defender. That is, the defender's decisions on $X_{i,k}$'s can be independent of the realization of $W_{i,k}$'s. On the other hand, the attacker can observe the defender's moves. Thus in general, $W_{i,k}$ may depend on the realization of both $X_{j,\tau}$ and $W_{j,\tau}$ for any $j$ and $\tau$ such that $T_j(\tau)<T_i(k)$ where $T_j(\tau)$ refers to the time instance of node $j$'s $\tau$-th defense.

\subsection{Defender's Problem}
We model the total cost to the defender as the summation of the total time when nodes are compromised and the total move cost. The defender aims at maximizing its payoff, which is defined as the negation of its total loss. Given the attacker's strategy $\{W_{i,k}\}$, the defender faces the following optimization problem:

\begin{equation}
\begin{aligned}
\label{generaldefender1}
\max_{\{X_{i,k}\},L_i}&E\bigg[ \sum_{i=1}^N\bigg(\frac{-\left(T-\sum_{k=1}^{L_i} \min(W_{i,k}+\alpha_{i,k},X_{i,k})\right)\cdot r_i}{T}\\
                        &\ \ \ \ \ \ \ -\frac{L_i C_i^D}{T}\bigg) \bigg] \\
\end{aligned}
\end{equation}
and the optimization variable $X_{i,k}$ and $L_i$ satisfy the following two constraints.
\begin{equation}\label{eqa:defenderconstraint}
\begin{aligned}
&\sum_{i=1}^N\frac{L_i}{T}\leq B \ \text{with probability 1} \ \\
&\sum_{k=1}^{L_i} X_{i,k}\leq T\ \text{with probability 1} \ \forall i
\end{aligned}
\end{equation}
where $L_i$ (a random variable) is the total number of defense applied to node $i$ during time $T$. In ~\eqref{generaldefender1}, $T-\sum_{k=1}^{L_i} \min(W_{i,k}+\alpha_{i,k},X_{i,k})$ refers to the total time when node $i$ is compromised and $L_iC_i^D$ is the overall move cost. The first constraint defines an upper bound $B$ of the average number of nodes that can be protected at any time. 
The second constraint in \eqref{eqa:defenderconstraint} defines the feasible set of $X_{i,k}$.

\subsection{Attacker's Problem}
Given the defender's strategy $\{X_{i,k}\}$, the total cost of attacking node $i$ is then $(\sum_{k=1}^{L_i}\phi({W_{i,k},X_{i,k}}))\cdot C_i^A$, where $\phi({W_{i,k},X_{i,k}})= 1$ if $W_{i,k}<X_{i,k}$ and $\phi({W_{i,k},X_{i,k}})= 0$ otherwise. It is important to note that when $W_{i,k}\geq X_{i,k}$, the attacker actually gives up its $k$-th attack against node $i$ (this is possible as the attacker can observe when the defender moves). 
The attacker's problem can be formulated as follows, where $M$ is an upper bound on the average number of nodes that the attacker can attack at any time instance.
\begin{equation}
\begin{aligned}
\label{generalattacker1}
\max_{W_{i,k}}&\ \ E\bigg[\sum_{i=1}^N \frac{(T-\sum_{k=1}^{L_i}\min(W_{i,k}+\alpha_{i,k},X_{i,k}))\cdot r_i}{T}\\
                    &\ \ \ \ \ \ \ \ \ \ -\frac{(\sum_{k=1}^{L_i}\phi({W_{i,k},X_{i,k}}))\cdot C_i^A}{T}\bigg]\\
\end{aligned}
\end{equation}
and the attacker needs to satisfy the following constraint
\begin{equation}\label{eqa:attackerconstraint}
E\left[\sum_{i=1}^N\frac{1}{T}\int_0^T v_i(t)dt\right]\leq M
\end{equation}
where $v_i(t)=1$ if the attacker is attacking node $i$ at time $t$ and $v_i(t)=0$ otherwise. Note that we make the assumption that the attacker has to keep consuming resources when the attack is in progress. 
We further have the following equation:
\begin{equation}
\label{generalattacker3}
\int_0^T v_i(t)dt=\sum_{k=1}^{L_i}\left(\min(W_{i,k}+\alpha_{i,k},X_{i,k})-\min(W_{i,k},X_{i,k})\right)
\end{equation}

Putting (\ref{generalattacker3}) into (\ref{generalattacker1}), \eqref{eqa:attackerconstraint} and moving the expectation inside, the attacker's problem becomes
\begin{equation}
\begin{aligned}
\label{generalattacker4}
\max_{W_{i,k}} &\sum_{i=1}^N \bigg[\frac{T\cdot r_i-E[\sum_{k=1}^{L_i}\min(W_{i,k}+ \alpha_{i,k},X_{i,k})]\cdot r_i}{T}\\
  &\ \ \ \ \ \ \ \ \ \ \ \ \ \ -\frac{E[\sum_{k=1}^{L_i}P(W_{i,k}<X_{i,k})]\cdot C_i^A}{T}\bigg]\\
\end{aligned}
\end{equation}
with resource constraints
\begin{equation}\label{eqa:resourceconstaint}
\begin{aligned}
\sum_{i=1}^N \bigg[&\frac{E[\sum_{k=1}^{L_i}(\min(W_{i,k}+\alpha_{i,k},X_{i,k})-\min(W_{i,k},X_{i,k}))]}{T}\bigg]\\
&\leq M
\end{aligned}
\end{equation}

\begin{table}[!t]
	\centering
	\caption{List of Notations}
	\renewcommand{\arraystretch}{1.05}
	\label{notation}
	\small{
		\begin{tabular}{|c|l|}
			\hline \multicolumn{1}{|l|}{Symbol} & \multicolumn{1}{|c|}{Meaning} \\
			\hline
			$T$ & time horizon\\
			\hline
			$N$ & number of nodes\\
			\hline
			$r_i$ & value per unit of time of compromising node $i$\\
			\hline
			$\alpha_{i,k}$ & attacking time for node $i$ in the $k$-th move\\
			\hline
			$C_i^A$ & attacker's move cost for node $i$\\
			\hline
			$C_i^D$ & defender's move cost for node $i$\\
			\hline
			$W_{i,k}$ & attacker's waiting time in its $k$-th move for node $i$\\
			\hline
			\multirow{2}{*}{$X_{i,k}$} & time between the $(k-1)$-th and the $k$-th defenses\\ &for node $i$\\
			\hline
			$B$ & budget to the defender, greater than 0\\
			\hline
			$M$ & budget to the attacker, greater than 0\\
			\hline
			$m_i$ & frequency of defenses for node $i$\\
			\hline
			\multirow{2}{*}{$p_i$} & probability of immediate attack on node $i$\\ &after it recovers\\
			\hline
			$L_i$ & the number of defense moves for node $i$\\
			\hline
		\end{tabular}
	}
\end{table}

\section{Best Responses}\label{sec:best-response}
In this section, we analyze the best-response strategies for both players. Our main result is that when the attacker employs an {\it i.i.d.} strategy, a periodic strategy is a best response for the defender, and vice versa. To prove this result, however, we have provided characterization of best responses in more general settings.
\subsection{Defender's Best Response}\label{sec:defender-best-response}
We first show that an optimal deterministic defense strategy is always optimal in general for \eqref{generaldefender1}. We then prove that the periodic defense is optimal against {\it i.i.d.} attacks.

\begin{lemma}
\label{lemma:general}
Suppose $x_{i,k}^\star$ and $l_i^\star$ are the optimal solutions of \eqref{generaldefender1} among all deterministic strategies, then they are also optimal among all the strategies (including both adaptive and non-adaptive strategies).
\end{lemma}
\begin{proof}
Consider a general defense strategy $X_{i,k}$, we define $x_{i,k}$ and $l_i$ as the realizations of $X_{i,k}$ and $L_i$ respectively and let $\mathcal{C}=\{(x_{i,k},\ l_i) | \sum_{i=1}^N \frac{l_i}{T}\leq B\ and\ \sum_{k=1}^{l_i}x_{i,k}\leq T\}$. Let $U^D(X_{i,k},L_i)$ denote the target function of \eqref{generaldefender1} and denote
\begin{equation}
\begin{aligned}
&\hat{U}^D(x_{i,k},l_i)\\
&=\sum_{i=1}^N\frac{-\left(T-\sum_{k=1}^{l_i} E[\min(W_{i,k}+\alpha_{i,k},x_{i,k})]\right)\cdot r_i-l_iC_i^D}{T}
\end{aligned}
\end{equation}
Since the defender cannot observe the attacker's behavior, any realization of $X_{i,k}$ can be pre-determined. Thus, we can compute the total expected payoff for defender as follows:
\begin{equation}
\begin{aligned}
&U^D(X_{i,k},L_i)\\
&=P(X_{i,k}=x_{i,k}^\star\ L_i= l_i^\star,\ \forall i,k)\cdot \hat{U}^D(x_{i,k}^\star,l_i^\star)\\
&+\sum_{(x_{i,k},\ l_i)\in \mathcal{C} \atop x_{i,k}\neq X_{i,k}^\star\ l_i\neq L_i^\star}P(X_{i,k}= x_{i,k}\ L_i=l_i,\ \forall i,k)\cdot \hat{U}^D(x_{i,k},l_i)\\
&\leq \hat{U}^D(x_{i,k}^\star,l_i^\star)
\end{aligned}
\end{equation}
The equality holds only when $X_{i,k}=x_{i,k}^\star\ L_i= l_i^\star\ \forall i,k\ w.p.1$. Therefore, $x_{i,k}^\star$ and $l_i^\star$ are also optimal among all the defender's strategies.
\end{proof}

According to the lemma, it suffices to consider defender's strategies where both $X_{i,k}$ and $L_{i,k}$ are deterministic. It is also worth mentioning that the order in which nodes are defended makes no difference since the nodes are independent of each other. We then define the set of {\it i.i.d.} attack strategies and show that periodic defense is a best response against {\it i.i.d.} attacks.

\begin{definition}
\label{definition2}
An attack strategy is called {\it i.i.d.} if it is non-adaptive, and $W_{i,k}$ is independent across $i$ and is {\it i.i.d.} across $k$.
\end{definition}

\begin{restatable}{theorem}{periodictheorem}
\label{theorem:defender1}
Periodic defense is a best response among all defense strategies if the attacker employs an {\it i.i.d.} strategy.
\end{restatable}
To prove this result, we need the following definition. 
\begin{definition}
\label{definition1}
For a given $L_i$, we define a set $\mathcal{X}_i$ that includes all deterministic defense strategies for node $i$ with the following properties:
\begin{enumerate}
  \item $ \sum_{k=1}^{L_i} X_{i,k}=T$;
  \item $ F_{W_{i,k}+\alpha_{i,k}}(X_{i,k})=F_{W_{i,j}+\alpha_{i,j}}(X_{i,j})\ \  \forall k,j$,
\end{enumerate}
where $F_{W_{i,k}+\alpha_{i,k}}(\cdot)$ is the marginal CDF of $W_{i,k}+\alpha_{i,k}$. Let $\mathcal{X}$ denote the set of defense strategies where for each node $i$, a strategy in $\mathcal{X}_i$ is adopted.
\end{definition}

Note that (1) $\mathcal{X}_i$ can be an empty set in general due to the randomness of $W_{i,k}+\alpha_{i,k}$; (2) for deterministic $X_{i,k}$, $W_{i,k}$ is independent of any $X_{j,\tau}$ s.t. $T_j(\tau)\geq T_i(k)$. The following lemma shows that when $\mathcal{X}_i$ is non-empty for all $i$, any strategy that belongs to $\mathcal{X}$ is a defender's best deterministic strategy against a non-adaptive attacker.

\begin{lemma}\label{lemma:defenderlemma1}
Consider a non-adaptiv attack strategy. For any given set of $\{L_i\}$ with $\sum_{i=1}^N\frac{L_i}{T}\leq B$, if $\mathcal{X}_i \neq \emptyset$ for any $i$, then any strategy 
in $\mathcal{X}$ is a best deterministic strategy for the defender.
\end{lemma}

\begin{proof}
We first define the defender's payoff for node $i$ as
\begin{equation}
\begin{aligned}
U_i^D(X_{i,k},L_i)=&\frac{-\bigg(T-\sum_{k=1}^{L_i} E[\min(W_{i,k}+\alpha_{i,k},X_{i,k})]\bigg)\cdot r_i}{T}\\
                    &\ \ -\frac{L_iC_i^D}{T}.
\end{aligned}
\end{equation}
Since $\{L_i\}$ are fixed, Problem~\eqref{generaldefender1} can be divided into $N$ independent sub-problems as follows:
\begin{equation}
\label{lemma4.2proof}
\begin{aligned}
&\max_{ X_{i,k}} U_i^D(X_{i,k})\\
 &\ \ s.t.\  \sum_{k=1}^{L_i} X_{i,k}\leq T
\end{aligned}
\end{equation}

We first assume that $F_{W_{i,k}+\alpha_{i,k}}(X_{i,k})$ 
is continuous for any $i$ and $k$.
Since the attacking strategy is non-adaptive,
$X_{i,k}$ is independent of $W_{i,k}$. We can then prove that the objective function is concave 
by showing that the Hessian matrix of $U_i^D(\{X_{i,k}\})$ with respect to $X_{i,k}, (1 \leq k \leq L_i)$ is negative sem-definite. We note that even when $F_{W_{i,k}+\alpha_{i,k}}(X_{i,k})$ is not continuous, the concavity can still be proved using the subgradient concept. The details are omitted to save space.

Since $U_i^D(X_{i,k})$ is concave and continuously differentiable, 
the KKT conditions are both sufficient and necessary. From the KKT conditions, we have $\nu^\star(\sum_{k=1}^{L_i} X_{i,k}-T)=0$ and $ F_{W_{i,k}+\alpha_{i,k}}(X_{i,k})=F_{W_{i,j}+\alpha_{i,k}}(X_{i,j}), \forall k,j$, where $\nu^\star$ is the Lagrangian multiplier. It is clear that $U_i^D(X_{i,k})$ is maximized when the constraint is tight, that is,  $\sum_{k=1}^{L_i}X_{i,k}=T$. Note that  there may exists a set of $X_{i,k}$ with $\sum_{k=1}^{L_i} X_{i,k}< T$ that is also optimal for~\eqref{lemma4.2proof}. Thus, the two conditions in Definition~\ref{definition1} are sufficient but not necessary.
\end{proof}

We now prove Theorem~\ref{theorem:defender1}.
\begin{proof}
For any fixed $\{L_i\}$, let $X_i \triangleq [\frac{T}{L_i} \frac{T}{L_i} \cdots \frac{T}{L_i}]$. It is easy to check that $\{X_i\}$ satisfies the fist property in Definition~\ref{definition1} and will satisfy the second property if $\alpha_{i,k}$ is {\it i.i.d.} with respect to $k$. According to Lemma~\ref{lemma:defenderlemma1}, $\{X_i\}$ is an optimal (deterministic) solution given $\{L_i\}$. It follows that if we let $\{L_i^\star\}$ denote the optimal solution of
\begin{displaymath}
\max_{L_i}\sum_{i=1}^N\frac{-\left(T-\sum_{k=1}^{L_i} E[\min(W_{i,k}+\alpha_{i,k},\frac{T}{L_i})]\right)\cdot r_i-L_i C_i^D}{T}
\end{displaymath}
with resource constraint $\sum_{i=1}^N\frac{L_i}{T}\leq B$. Then $X_i^{\star} \triangleq [\frac{T}{L_i^{\star}} \frac{T}{L_i^{\star}} \cdots \frac{T}{L_i^{\star}}]$ is an optimal solution to the defender's problem. Hence, a periodic strategy with periods of $X_i^{\star}$ for all $i$ is a best-response strategy for the defender. 
\end{proof}

According to Theorem~\ref{theorem:defender1}, the defender use periodic strategy to keep the system stable, in the sense of the same total loss between two defenses. Since the distribution of attacker's waiting time $W_{i,k}$ does not change with time, a fixed defense interval provides the same expected payoff between every two consecutive moves. Moreover, the convexity of the defender's optimization problem guarantees an optimal solution under a given attack strategy.

\subsection{Attacker's Best Response}
We first analyze the attacker's best response against any deterministic defense strategy, then show that the {\it i.i.d.} strategy is the best response against periodic defense.

\begin{definition}
\label{definition3}
An attack strategy is called independent non-adaptive if it is non-adaptive, and $W_{i,k}$ is independent across $i$ and $k$.
\end{definition}
\begin{lemma}\label{lemma:adaptiveattacker}
When the defense strategy is deterministic, for any attacking
strategy (adaptive or non-adaptive), there always exists an independent non-adaptive strategy that gives the attacker the same payoff.
\end{lemma}
\begin{proof}
When the defense strategies are deterministic, we can move the expectation in \eqref{generalattacker4} after the summation over $k$ and the expectation is with respect to $\alpha_{i,k}$. The same for constraint \eqref{eqa:resourceconstaint}.
Then, the proof is done as long as we can construct an independent non-adaptive strategy $W'_{i,k}$ such that for all $i$ and $k$, we have
\begin{enumerate}
\item $E[\min(W_{i,k}+\alpha_{i,k},X_{i,k})]=E[\min(W'_{i,k}+\alpha_{i,k},X_{i,k})]$;
\item $E[\min(W_{i,k},X_{i,k})]=E[\min(W'_{i,k},X_{i,k})]$;
\item $P(W_{i,k}<X_{i,k})=P(W'_{i,k}<X_{i,k})$.
\end{enumerate}
Since $X_{i,k}$ is deterministic and $\alpha_{i,k}$ is independent across $i$ and $k$, the expectation above is with respect to the marginal distribution of $W_{i,k}$ only. Thus, we can construct $W'_{i,k}$ whose distribution is the same as $W_{i,k}$'s marginal distribution which does not depend on any realization of $X_{j,\tau}$ and $W_{j,\tau}$ s.t. $T_j(\tau)<T_i(k)$. Meanwhile, $W'_{i,k}$ is independent across $i$ and $k$.
\end{proof}
According to Lemma~\ref{lemma:adaptiveattacker}, it suffices to consider independent non-adaptive strategies when the defender uses deterministic strategies.

\begin{restatable}{lemma}{primelemmadefender}
\label{lemma:attackerlemma2}
When the defense strategy is deterministic, the attacker's best response (among non-adaptive strategies) must satisfy the following condition
\begin{equation}
\label{lemma:structure1}
W_{i,k}^\star=\begin{cases}
0\ \ \ &w.p.\ p_{i,k}\\
\geq X_{i,k}\ \ \ &w.p.\ 1-p_{i,k}
\end{cases}
\end{equation}
\end{restatable}
Please find the proof in Section~\ref{subsec::proofofattacklemma}. Lemma~\ref{lemma:attackerlemma2} implies that for each node $i$, the attacker's best strategy is to either attack node $i$ immediately after it realizes the node's recovery, or gives up the attack until the defender's next move. There is no incentive for the attacker to wait a small amount of time to attack a node before the defender's next move. The constraint $M$ actually determines the probability that the attacker will attack immediately. If $M$ is large enough, the attacker will never wait after defender's each move. We then find the attacker's best response when the defender employs the periodic strategy.

\begin{theorem}
\label{theorem:attacker1}
Assume that for any $i$, the attacking times $\alpha_{i,k}$'s are {\it i.i.d.} across $k$. When the defender employs a periodic strategy, the {\it i.i.d.} strategy is the attacker's best response among all strategies.
\end{theorem}
\begin{proof}
Suppose that the defender uses a periodic strategy where for any $i$, $X_{i,k} = 1/m_i$ for any $k$. With~\eqref{lemma:structure1}, the attacker's problem \eqref{generalattacker4} can be simplied to a fractional knapsack problem with decision variables $\{p_{i,k}\}$. For any given node $i$, $p_{i,k}$'s unit reward (payoff in the target function divided by weight in the constraint) across $k$ are all equal when $\alpha_{i,k}$'s are {\it i.i.d.} across $k$. Thus, setting all the $p_{i,k}$ in \eqref{lemma:structure1} equal is one of the optimal solution. Therefore, the {\it i.i.d.} strategy is a best solution for attacker when the defender uses a periodic strategy.
\end{proof}

\subsection{Simplified Optimization Problems}\label{sec:simplified}
We put particular emphasis on the case where the defender employs a periodic strategy and the attacker uses an {\it i.i.d.} strategy.
According to Theorem~\ref{theorem:defender1} and Theorem~\ref{theorem:attacker1}, periodic defense and {\it i.i.d.} attack can form a pair of best-response strategies with respect to each other. Consider such pairs of strategies. Let $m_i \triangleq \frac{L_i}{T}=\frac{1}{X_{i,k}}$, and let $p_i$ denote the probability that $W_{i,k}=0$ for all $k$. We assume that all the attacking times $\alpha_{i,k}$ are {\it i.i.d.} across $k$ and omit the subscript $k$ in $\alpha_{i,k}$. The optimization problems to the defender and the attacker can then be simplified as follows.

\noindent Defender's problem:
\begin{equation}
\begin{aligned}
\label{staticdefender1}
\max_{m_i}&\sum_{i=1}^N \left[\left(E[\min{(\alpha_{i},\frac{1}{m_i})}]p_ir_i-C_i^D\right)\cdot m_i-p_ir_i\right]\\
 &s.t.\ \sum_{i=1}^N m_i\leq B
\end{aligned}
\end{equation}

\noindent Attacker's problem:
\begin{equation}
\label{staticattacker2}
\begin{aligned}
\max_{p_i}&\sum_{i=1}^N p_i\cdot \left( r_i(1-E[\min(\alpha_{i},\frac{1}{m_i})]\cdot m_i)-C_i^Am_i\right)\\
s.t.& \sum_{i=1}^N E[\min(\alpha_i,\frac{1}{m_i})]\cdot m_i \cdot p_i\leq M
\end{aligned}
\end{equation}

We observe that the defender's problem is a continuous convex optimization problem, 
while the attacker's problem is a fractional knapsack problem. Therefore, the best response strategy of each side can be easily determined. Also, the time period $T$ disappears in both problems. It is worth mentioning that finding the Nash Equilibrium of \eqref{staticdefender1} - \eqref{staticattacker2} is very challenging since the constraint of \eqref{staticattacker2} is non-convex with respect to $m_i$, thus the strategy space of this generalized Nash Equilibrium problem (GNEP) is not jointly convex. 

\subsection{Markovian Strategies}\label{sec:markov}
Based on Theorems~\ref{theorem:defender1} and \ref{theorem:attacker1}, the defender's periodic strategy and attacker's {\it i.i.d.} strategy form a Nash equilibrium among all adaptive and non-adaptive strategies. However, it remains unclear what is the best response if one of the players uses an adaptive strategy. To the best of our knowledge, there has been virtually no discussion about adaptive strategies in the field of stealthy attacks. Further, even though 
a deterministic strategy is always optimal for the defender based on Lemma~\ref{lemma:general}, there may still exist 
non-deterministic strategies that are also optimal. Meanwhile, Nash equilibria under 
more general strategies from both players may exist. In this section, we provide some preliminary results in this direction by considering Markovian strategies from both the defender's and the attacker's perspectives. We assume that the attacker's waiting times $W_{i,k}$ follow~\eqref{lemma:structure1} and define a Markovian attacking strategy as follows:

\begin{definition}\label{def:markovstrategy}
An attacking strategy is a {\bf Markovian strategy} if the attack probabilities $\{p_{i,k}\}$ for node $i$ follow a discrete Markov chain over $K$ states $v_1,v_2,\cdots,v_K$ with transition matrix $M^A_i$. That is, $\Pr(p_{i,k+1}=v_s | p_{i,k}=v_t)=M^A_i(s,t)$ for any $s$ and $t$. 
\end{definition}

A Markovian defense strategy is defined similarly by considering $\{X_{i,k}\}$ instead of $\{p_{i,k}\}$. For tractability, we only consider the expected payoffs for the defender in a steady state. We show our main results about Markovian strategies in the following.

\begin{theorem}\label{thm:markovattacker}
If the attacker employs an ergodic Markovian strategy, the periodic strategy is defender's best response.
\end{theorem}

The proof can be found in Section~\ref{subsec:markovattakcerproof}. Theorem~\ref{thm:markovattacker} tells us that the defender still prefers using a periodic strategy when the attacker's strategy space includes Markovian strategies. Consequently, the pair of periodic strategy and {\it i.i.d.} strategy naturally forms the Nash equilibrium in this case. However, the {\it i.i.d.} attack strategy may not be optimal against a Markovian defending strategy as shown in the following theorem.

\begin{theorem}\label{thm:defendermarkov}
If the defender employs a Markovian strategy, the {\it i.i.d.} attack strategy is not optimal in general.
\end{theorem}


The detailed proof can be found in Section~\ref{subsec:markovproof}. Theorem~\ref{thm:defendermarkov} tells us that the attacker may use an adaptive strategy against the Markovian defending strategy. Compared to the defender, since the attacker is able to observe the defending periods and the node states, the attacking strategy may become state-dependent.  Therefore, 
Nash equilibria beyond periodic defense and {\it i.i.d.} attack can exist in the space of both adaptive and non-adaptive strategies.

\subsection{Discussion on Security Games in Networks}
In this work, we focus on protecting a set of independent nodes where the payoff functions are additive, that is, the total payoff to a player is a weighted summation of the payoffs from each node. Even in this case, finding the equilibrium solutions of the game \eqref{staticdefender1} - \eqref{staticattacker2} is already very challenging as we mentioned in Section~\ref{sec:simplified}. Solving a security game in a general network setting that yields non-additive utility is even harder. Because of that, existing security game work typically assume additive utility as we did.  

To extend our solutions discussed in Sections~\ref{sec:equilibrium} and~\ref{sec:sequential} to a network setting, a promising direction is to introduce non-additive payoff functions to the defender and the attacker to capture the dependencies of node values.   There are several recent work \cite{gueye2012towards,wang2017security,wang2017non} that consider security games in network settings. In particular, Wang et al. \cite{wang2017security} developed a general framework to convert a security game with non-additive utility to a combinatorial optimization problem over a set system, and characterized the complexity of finding the Nash Equilibrium. However, efficient algorithms are only known for some special cases and none of them apply to our setting directly. Further, most previous work on security games including \cite{wang2017security} consider a static setting (or the stead state in a repeated setting) where the game is played only once, which cannot faithfully model the joint spatial and temporal decisions in dynamic stealthy games as we consider in this paper.


\section{Nash Equilibria}\label{sec:equilibrium}

In this section, we study the set of Nash Equilibria of the game where the defender employs a periodic strategy, and the attacker employs an {\it i.i.d.} strategy. For tractability, we further assume that the attacking time $\alpha_{i,k}$ is deterministic for all $i$ and we omit the subscript $k$. 
We show that this game always has a Nash equilibrium and may have multiple equilibria of different values.


We first observe that for deterministic $\alpha_i$, when $m_i \geq \frac{1}{\alpha_i}$, the defender's payoff becomes $-m_iC^D_i$, which is maximized when $m_i = \frac{1}{\alpha_i}$. Therefore, it suffices to consider $m_i \leq \frac{1}{\alpha_i}$. Thus, the optimization problems to the defender \eqref{staticdefender1} and the attacker \eqref{staticattacker2} can be simplified as follows.

For a given $p$, the defender aims at maximizing its payoff:
\begin{equation}
\begin{aligned}
\label{sbgame1}
\max_{m_i}&\sum_{i=1}^N[m_i(r_i\alpha_ip_i-C_i^D)-p_ir_i]\\
s.t.\ \ &\sum_{i=1}^N m_i\leq B\\
&0\leq m_i\leq \frac{1}{\alpha_i}, \forall i
\end{aligned}
\end{equation}

On the other hand, for a given $m$, the attacker aims at maximizing its payoff:
\begin{equation}
\begin{aligned}
\label{sbgame2}
\max_{p_i}&\sum_{i=1}^Np_i[r_i-m_i(r_i\alpha_i+C_i^A)]\\
s.t.\ \ &\sum_{i=1}^N m_i\alpha_ip_i\leq M\\
&0\leq p_i\leq 1, \forall i
\end{aligned}
\end{equation}

For a pair of strategies $(m,p)$, the payoff to the defender is $U_d(m,p) = \sum_{i=1}^N[m_i(p_ir_iw_i-C_i^D)-p_ir_i]$, while the payoff to the attacker is $U_a(m,p) = \sum_{i=1}^Np_i[r_i-m_i(r_iw_i+C_i^A)]$. A pair of strategies $(m^*,p^*)$ is called a (pure strategy) {\it Nash Equilibrium (NE)} if for any pair of strategies $(m,p)$, we have $U_d(m^*,p^*) \geq U_d(m,p^*)$ and $U_a(m^*,p^*) \geq U_a(m^*,p)$. In the following, we assume that $C^A_i>0$ and $C^D_i>0$. The cases where $C^A_i=0$ or $C^D_i=0$ or both exhibit slightly different structures, but can be analyzed using the same approach. Without loss of generality, we assume $r_i>0$ and $\frac{C^D_i}{r_iw_i}\leq 1$ for all $i$. Note that if $r_i = 0$, then node $i$ can be safely excluded from the game, while if $\frac{C^D_i}{r_iw_i}> 1$, the coefficient of $m_i$ in $U_d$ (defined below) is always negative and there is no need to protect node $i$.

Let $\mu_i(p) \triangleq p_ir_iw_i-C_i^D$ denote the coefficient of $m_i$ in $U_d$, and $\rho_i(m) \triangleq \frac{r_i-m_i(r_iw_i+C_i^A)}{m_iw_i}$. Note that for a given $p$, the defender tends to protect more a component with higher $\mu_i(p)$, while for a given $m$, the attacker will attack a component more frequently with higher $\rho_i(m)$. When $m$ and $p$ are clear from the context, we simply let $\mu_i$ and $\rho_i$ denote $\mu_i(p)$ and $\rho_i(m)$, respectively.

To find the set of NEs of our game, a key observation is that if there is a full allocation of defense budget $B$ to $m$ such that $\rho_i(m)$ is a constant for all $i$, any full allocation of the attack budget $M$ gives the attacker the same payoff. Among these allocations, if there is further an assignment of $p$ such that $\mu_i(p)$ is a constant for all $i$, then the defender also has no incentive to deviate from $m$; hence $(m,p)$ forms an NE. The main challenge, however, is that such an assignment of $p$ does not always exist for the whole set of nodes. Moreover, there are NEs that do not fully utilize the defense or attack budget as we show below. To characterize the set of NEs, we first prove the following properties satisfied by any NE of the game. For a given strategy $(m,p)$, we define $\mu^*(p) \triangleq \max_i \mu_i(p)$, $\rho^*(m) \triangleq \min_i \rho_i(m)$, $F(p) \triangleq \{i: \mu_i(p) = \mu^*(p)\}$, and $D(m,p) \triangleq \{i \in F: \rho_i(m) = \rho^*(m)\}$. We omit $m$ and $p$ when they are clear from the context. 

\begin{lemma}\label{lemma:range}
In any NE, (1) $m_i \leq \frac{r_i}{r_iw_i+C^A_i}$ and (2) $p_i \geq \frac{C^D_i}{r_iw_i}$.
\end{lemma}

\begin{proof}
To prove the first property, suppose $m_i > \frac{r_i}{r_iw_i+C^A_i}$. Then $p_i$ must be 0; otherwise the benefit for attacking $i$ becomes negative. This in turn implies that $m_i = 0$ by the assumption that $C^D_i>0$, a contradiction. To prove the second property, suppose $p_i < \frac{C^D_i}{r_iw_i}$. Then we have $\mu_i < 0$, which implies $m_i = 0$ and therefore $p_i = 1$ since $r_i > 0$, a contradiction.
\end{proof}

\begin{lemma}\label{lemma:necessary}
If $(m,p)$ is an NE, we have (see Table~\ref{table_NE}):
\begin{enumerate}
\item $\forall i \not \in F, m_i = 0, p_i = 1, \rho_i = \infty$;
\item $\forall i \in F \backslash D, m_i \in [0,\frac{r_i}{w_ir_i+C^A_i}], p_i = 1$;
\item $\forall i \in D, m_i \in [0,\frac{r_i}{w_ir_i+C^A_i}], p_i \in [\frac{C^D_i}{r_iw_i}, 1]$.
\end{enumerate}
\end{lemma}
\begin{proof}
We first show that if $m_i >0$ and $m_j >0$, then $\mu_i = \mu_j$. Suppose $\mu_i < \mu_j$. Then it is better to protect $i$ than protecting $j$. Since $m_i>0$, we must have $m_j =  \frac{1}{w_j}>\frac{r_i}{r_iw_i+C^A_i}$ by the assumption that $C^A_i > 0$, a contradiction. It follows that $m_i = 0$ $\forall i \not \in F$ and $m_i \in [0,\frac{r_i}{r_iw_i+C^A_i}]$ $\forall i \in F$. Since when $m_i = 0$, we must have $\rho_i = \infty$, and $p_i = 1$, $p_i = 1, \rho_i = \infty$ $\forall i \not \in F$. It remains to show that $p_i = 1$ for all $i \in F \backslash D$. Assuming $F \backslash D \neq \emptyset$, then we have $\rho_j < \infty$ for $j \in D$, which implies that $m_j > 0$ for $j \in D$. Since $\rho_i < \rho^*$ for $i \in F \backslash D$, it is more beneficial to attack $i$ that any $j \in D$. Since $p_j > 0$ and $m_j > 0$ for $j \in D$, we must have $p_i = 1$.
\end{proof}

\begin{lemma}\label{lemma:order}
If $(m,p)$ forms an NE, then for $i \in D, j \in F \backslash D$ and $k \not \in
F$, we have $r_iw_i - C^D_i \geq r_jw_j - C^D_j > r_kw_k - C^D_k$.
\end{lemma}
\begin{proof}
Since $\mu_i = \mu_j$ for $i \in D$ and $j \in F \backslash D$ by the definitions of $F$ and $D$, and $p_i \leq p_j = 1$ by
Lemma~\ref{lemma:necessary}, we have $r_iw_i - C^D_i \geq \mu_i = \mu_j \geq r_jw_j - C^D_j$. On the other hand, since $\mu_j > \mu_k$ by the
definition of $F$, and $p_j = p_k = 1$ by Lemma~\ref{lemma:necessary}, we have $r_jw_j - C^D_j = \mu_j > \mu_k = r_kw_k - C^D_k$.
\end{proof}

According to the above lemma, to find all the equilibria of the game, it suffices to sort all the nodes by a non-increasing order of $r_iw_i-C^D_i$, and consider each $F_h$ consisting of the first $h$ nodes such that $r_hw_h-C^D_h > r_{h+1}w_{h+1}-C^D_{h+1}$, and each subset $D_k \subseteq F_h$ consisting of the first $k \leq h$ nodes in the list. In the following, we assume such an ordering of nodes. Consider a given pair of $F$ and $D \subseteq F$. By Lemma~\ref{lemma:necessary} and the definitions of $F$ and $D$, the following conditions are satisfied by any NE with $F(p) = F$ and $D(m,p) = D$.
\begin{align}
m_i = 0, p_i = 1, \forall i \not \in F; \label{C1} \\
m_i \in [0,\frac{r_i}{w_ir_i+C^A_i}], p_i = 1, \forall i \in F \backslash D;  \label{C2} \\
m_i \in [0,\frac{r_i}{w_ir_i+C^A_i}], p_i \in [\frac{C^D_i}{r_iw_i}, 1], \forall i \in D; \label{C3}\\
\sum_{i \in F} m_i \leq B, \sum_{i \in F} m_iw_ip_i \leq M; \label{C4}\\
\mu_i = \mu^*, \forall i \in F;\ \ \ \ \ \mu_i < \mu^*, \forall i \not \in F;\label{C5}\\
\rho_i = \rho^*, \forall i \in D;\ \ \ \ \ \rho_i > \rho^*, \forall i \not \in D.\label{C6}
\end{align}

\begin{table}
\centering
\caption{Necessary Conditions for NEs}\label{table_NE}
\renewcommand{\arraystretch}{1.2}
\small{
\begin{tabular}{|c|c|c|c|}
\hline $i \in $ & $D$ & $F \backslash D$ & $\overline{F}$ \\
\hline $m_i$ & $[0,\frac{r_i}{w_ir_i+C^A_i}]$ & $[0,\frac{r_i}{w_ir_i+C^A_i}]$ & 0 \\
\hline $p_i$ & $[\frac{C^D_i}{r_iw_i}, 1]$ & 1 & 1 \\
\hline $\mu_i$ &  $\mu^*$ & $\mu^*$ & $<\mu^*$ \\
\hline $\rho_i$ & $\rho^*$ & $>\rho^*$ & $+\infty$ \\
\hline
\end{tabular}
}
\end{table}

The following theorem provides a full characterization of the set of NEs of the game.
\begin{theorem}
Any pair of strategies $(m,p)$ with $F(p) = F$ and $D(m,p) = D$ is an NE iff it is a solution to one of the following sets of constraints in addition to \eqref{C1} to \eqref{C6}.
\begin{enumerate}
\item $\sum_{i \in F} m_i = B$; $\rho^* = 0$;
\item $\sum_{i \in F} m_i = B$; $\rho^* > 0$; $\sum_{i \in F} m_iw_ip_i = M$;
\item $\sum_{i \in F} m_i = B$; $\rho^* > 0$; $p_i = 1, \forall i \in F$;
\item $\sum_{i \in F} m_i < B$; $\mu^* = 0$; $F = F_N$; $\rho^*=0$;
\item $\sum_{i \in F} m_i < B$; $\mu^* = 0$; $F = F_N$; $\rho^*>0$; $\sum_{i \in F} m_iw_ip_i = M$;
\item $\sum_{i \in F} m_i < B$; $\mu^* = 0$; $F = F_N$; $\rho^*>0$; $p_i = 1, \forall i \in F$.
\end{enumerate}
\end{theorem}
\begin{proof}
We first consider the cases when the budget constraint of the defender is tight, i.e., $\sum_{i \in F} m_i = B$ (cases 1-3). Since $m_i \leq \frac{r_i}{w_ir_i+C^A_i}$ in any NE by Lemma~\ref{lemma:range} and $m_i = 0$ for $i$ not in $F$ by Lemma~\ref{lemma:necessary}, we must have $B \leq \sum_{i \in F} \frac{r_i}{w_ir_i+C^A_i}$ in any NE. If $B = \sum_{i \in F} \frac{r_i}{w_ir_i+C^A_i}$, we have $\rho^* = 0$ (\textbf{case 1}). Assume $B < \sum_{i \in F} \frac{r_i}{w_ir_i+C^A_i}$. First consider the case $D = F$. We then have $m_i \leq \frac{r_i}{w_ir_i+C^A_i}, i \in F$. Hence, $\rho^*>0$ since $B < \sum_{i \in F} \frac{r_i}{w_ir_i+C^A_i}$. It follows that $\sum_{i \in F} m_iw_ip_i = M$ (\textbf{case 2}) unless $p_i = 1, \forall i \in F$ (\textbf{case 3}); otherwise, some $p_i, i \in F$ can be increased to get more benefit. Note that case 3 can happen only if $r_iw_i - C^D_i$ is the same for all $i \in F$. Next consider the case $D \subsetneq F$. If $B \in [\sum_{i \in E} \frac{r_i}{w_ir_i+C^A_i}$, $\sum_{i \in F} \frac{r_i}{w_ir_i+C^A_i})$, we again have $\rho^* = 0$ and get case 1, but with extra constraints regarding $i \in F \backslash D$ as required by \eqref{C2} and \eqref{C8}. Otherwise, if $B < \sum_{i \in D} \frac{r_i}{w_ir_i+C^A_i}$, by applying a similar argument as above, we again have $\rho^*>0$ and get case 2 or case 3 depending on whether the attacker's budget constraint is tight or not.

We next consider the cases when $\sum_{i \in F} m_i < B$ (cases 4-6). We first observe that $p_i = \frac{C^D_i}{r_iw_i}, \forall i \in F$, or equivalently, $\mu^* = 0$. Otherwise, if $\mu_i > 0$, $m_i$ can be further increased to reduce the cost due to the fact that $m_i \leq \frac{r_i}{r_iw_i+C^A_i} < \frac{1}{w_i}$ in any NE (by Lemma~\ref{lemma:range} and the assumption that $C^A_i>0$), a contradiction. We then have $F = F_N$ by its definition. Cases 4-6 then follow from a similar argument for cases 1-3 by distinguishing different values of $\rho^*$.
\end{proof}

In the following, NEs that fall into each of the six cases considered above are named as Type 1 - Type 6 NEs, respectively.  The next theorem shows that our game has at least one equilibrium and may have more than one NE.
\begin{theorem}
The attacker-defender game always has a pure strategy Nash Equilibrium, and may have more than one NE of different payoffs to the defender.
\end{theorem}

\begin{proof}
To show the first part, for any given index $h \leq N$, we define a pair of strategies $(m^h, p^h)$ as follows. Let $m^h_i = 0, \forall i > h$ and let $\{m^h_i, i \leq h\}$ be the solution to the constraints (1) $\sum_{i \leq h} m^h_i = B$ and (2) $\rho_i$ is a constant for all $i \leq h$; $p^h_i = 1, \forall i > h$ (hence $\mu_{h+1} = r_{h+1}w_{h+1}-C^D_h$), and $\forall i \leq h, p^h_i = \frac{\mu_{h+1}+C^D_i}{r_iw_i}$ if $h < N$, $p^h_i = 0$ otherwise.

We first prove the following claim. For a given $h$, let $h' \leq h$ denote the smallest index such that $r_{h'}w_{h'}-C^D_{h'} = r_hw_h-C^D_h$. Consider two pairs of strategies $(m^h,p^h)$ and $(m^{h'},p^{h'})$. We claim that if $\sum_{i \leq h}m^h_iw_ip^h_i < M$ and $\sum_{i \leq h}m^{h'}_iw_ip^{h'}_i \geq M$, then there is a Type 2 NE respecting $F_h$. Note that by definition, $\sum_{i \leq h}m^h_iw_ip^h_i < M$ is always true when $h = N$.

To prove the claim, we consider another pair of strategies $(m^h, p^{h'})$. If we have $\sum_{i \leq h}m^h_iw_ip^{h'}_i \geq M$, then since $\sum_{i \leq h}m^h_iw_ip^h_i < M$, there must exist $p$ with $p_i = 1, \forall i > h$, $p_h \in [\frac{\mu_{h+1}+C^D_h}{r_hw_h},1]$, and $p_i = \frac{\mu_h+C^D_i}{r_iw_i}, \forall i \leq h$ such that $\sum_{i \leq h}m^h_iw_ip_i = M$. Hence, $(m^h,p)$ is a Type 2 NE. On the other hand, if $\sum_{i \leq h}m^h_iw_ip^{h'}_i < M$, then since $\sum_{i \leq h}m^{h'}_iw_ip^{h'}_i \geq M$, there must exist $m$ with $m_i = 0, \forall i>h$, $m_i \in [0,m^h_i], \forall h' \leq i \leq h$, and $\{m^h_i, i \leq h\}$ be the solution to the constraints (1) $\sum_{i \leq h} m^h_i = B$ and (2) $\rho_i$ is a constant for all $i < h'$, such that $\sum_{i \leq h}m_iw_ip^{h'}_i = M$. We again get a Type 2 NE.

We then prove the theorem. First note that if $B \geq \sum_{i \leq N} \frac{r_i}{w_ir_i+C^A_i}$, then there is a Type 1 or Type 4 NE in $F_N$. Assume $B < \sum_{i \leq N} \frac{r_i}{w_ir_i+C^A_i}$. There is $h<N$ such that $B < \sum_{i \leq h} \frac{r_i}{w_ir_i+C^A_i}$ and $B \geq \sum_{i < h'} \frac{r_i}{w_ir_i+C^A_i}$, where $h'$ is defined as above. If there is an NE with respect to some $F(h''), h'' > h$, we are done. Otherwise, we have $\sum_{i \leq h}m^h_iw_ip^h_i < M$ by the claim. If $\sum_{i \leq h}m^h_iw_ip^{h'}_i \geq M$, there is a Type 2 NE as proved above. Otherwise, consider the pair of strategies $(m', p^{h'})$ where $m'_i = 0, \forall i > h$, $m_i = \frac{r_i}{w_ir_i+C^A_i}, \forall i < h'$, and $\{m^h_i, i \leq h\}$ is the solution to the constraints (1) $\sum_{i \leq h} m^h_i = B$ and (2) $\rho_i$ is a constant for all $i < h'$. If $\sum_{i \leq h}m'_iw_ip^{h'}_i \geq M$, there is Type 2 NE. Otherwise, there must be a Type 1 NE.

To show the second part, consider the following example with two nodes where $r_1 = r_2 = 1, w_1 = 2, w_2 = 1, C^D_1 = 1/5, C^D_2 = 4/5, C^A_1 = 1, C^A_2 = 7/2, B = 1/3$, and $M = 1/5$. It is easy to check that $m = (1/6,1/6)$ and $p = (3/20,9/10)$ is a Type 2 NE, and $m = (1/3,0)$ and $p = (p_1,1)$ with $p_1 \in [1/5,3/10]$ are all Type 1 NEs, and all these NEs have different payoffs to the defender.
\end{proof}

\section{Sequential Game}\label{sec:sequential}
In this section, we study the subgame perfect equilibrium~\cite{game-theory-Osborne} of the Stackelberg game when the defender employs a periodic strategy and the attacker employs an {\it i.i.d.} strategy. 
In the sequential game, the defender first commits to a strategy and makes it public, the attacker then responds accordingly. 
We assume that at $t=0$, the leader (defender) has determined its strategy and the follower (attacker) has learned the defender's strategy and determined its own strategy in response. In addition, the players do not change their strategies thereafter. Our objective is to identify the best sequential strategy for the defender.
We adopt the same assumption in Section~\ref{sec:equilibrium} and then define the subgame perfect equilibrium as follows:

\begin{definition}
\label{def:sequential}
A pair of strategies $(m^{\star},p^{\star})$ is a subgame perfect equilibrium of the sequential game 
if $m^{\star}$ is the optimal solution of
\begin{equation}
\begin{aligned}
\label{snegame1}
\max_{m_i}&\sum_{i=1}^N[m_i(r_i\alpha_ip^{\star}_i-C_i^D)-p^{\star}_ir_i]\\
s.t.\ \ &\sum_{i=1}^N m_i\leq B\\
&0\leq m_i\leq \frac{1}{\alpha_i}, \forall i
\end{aligned}
\end{equation}
where $p^{\star}_i$ is the optimal solution of
\begin{equation}
\begin{aligned}
\label{snegame2}
\max_{p_i}&\sum_{i=1}^Np_i[r_i-m_i(r_i\alpha_i+C_i^A)]\\
s.t.\ \ &\sum_{i=1}^N m_i\alpha_ip_i\leq M\\
&0\leq p_i\leq 1, \forall i
\end{aligned}
\end{equation}
\end{definition}

Note that in a subgame perfect equilibrium, $p^{\star}_i$ is the optimal solution of (\ref{snegame2}), but the defender's best strategy $m^{\star}_i$ is not necessarily optimal with respect to (\ref{snegame1}).
Due to the multi-node setting and the resource constraints, it is very challenging to identify an exact subgame perfect equilibrium strategy for the defender. We first establish several properties about the optimal defense strategy and then propose a dynamic programming based algorithm that finds a nearly optimal defense strategy.


To clearly state the properties, we partion all the nodes into four disjoint sets defined below:

\begin{enumerate}
\item $F=\{i |m_i > 0,\ p_i = 1\}$
\item $D=\{i |m_i > 0,\ 0 < p_i < 1\}$;
\item $E=\{i |m_i > 0,\ p_i = 0\}$;
\item $G=\{i |m_i = 0,\ p_i = 1\}$.
\end{enumerate}
\ifreport
\begin{proof}
It is obvious that $F$, $D$, $E$ and $G$ are disjoint. The three properties follow directly from the structure of the optimal solution to the attacker's problem and the remark made above.
\end{proof}
\fi

We observe that the set $D$ has at most one element since \eqref{snegame2} is a fractional knapsack problem. Let $\rho_i(m_i) \triangleq \frac{r_i-m_i(r_i\alpha_i+C_i^A)}{m_i\alpha_i}$. We use $m_d$ to represent $m_i,\ i\in D$ for simplicity and denote $\rho_d=\rho_d(m_d)$. If $D$ is empty, we pick any node $i$ in $F$ with minimum $\rho_i$ and treat it as a node in $D$.

\begin{lemma}\label{remark:rho}
For all optimal solutions of \eqref{snegame1}-\eqref{snegame2}, we always have $\rho_d\geq 0$
\end{lemma}

\begin{proof}
If $\rho_d<0$, the defender can give a smaller budget to the corresponding node to bring $\rho_d$ down to 0. In any case, the payoffs from nodes in sets $D$ and $E$ are $0$ since the attacker will give up attacking the nodes in sets $D$ and $E$. Thus, the defender has more budget to defend the nodes in sets $F$ and $G$ which brings more payoff. Therefore $\rho_d$ is always greater than or equal to $0$.
\end{proof}

Based on Lemma~\ref{remark:rho}, we only consider non-negative $\rho_d$ in the our analysis and algorithm.

\begin{lemma}\label{lemma:rho}
For any given nonnegative $\rho_d$, an optimal solution  for (\ref{snegame1})-(\ref{snegame2}) satisfies the following properties:
\begin{enumerate}
\item $r_i\alpha_i-C_i^D>0\ \forall i\in F\cup E\cup D$
\item $m_i\leq \overline{m}_i\ \forall i\in F$
\item $m_j=\overline{m}_j\ \forall j\in E$
\item $\overline{m}_i\leq \frac{1}{\alpha_i}\ \forall i$
\item $B-\sum_{i\in E}\overline{m}_i-m_d>0$.
\end{enumerate}
where $\overline{m}_i=\rho^{-1}_i (\rho_d)$
\end{lemma}

\begin{proof}
If $r_i\alpha_i-C_i^D \leq 0\ \forall i\in F\cup E\cup D$. there is no point for the defender to defend such node which will only make the payoff even worse due to high defending cost. Thus, all the nodes whose $r_i\alpha_i-C_i^D\leq 0$ are only in set $G$. For $\forall i\in F$, $\rho_i(\overline{m}_i)=\rho_d$ and $\rho_i(m_i)\geq \rho_d$. According to the reverse relationship between $\rho$ and $m_i$, we have $m_i\leq \overline{m}_i$. For $\forall j\in E$, since $\rho_j(\overline{m}_j)=\rho_d$ and $\rho_j(m_j)\leq \rho_d$, $\overline{m}_j$ is actually a lower bound for $m_j$. Setting $m_j=\overline{m}_j$ makes the cost from node $i$, which is $m_iC_i^D$ gets its minimum and so does the whole problem since it also uses the minimum budget from $B$. Therefore, more budget can be allocated for $m_i\ i\in F$ to minimize the cost from the nodes in set $F$. Further, it's easy to check $\overline{m}_i$ is always less than $\frac{1}{\alpha_i}$ for any given nonnegative $\rho_d$. As to the 5th property, if $B-\sum_{i\in E}\overline{m}_i-m_d\leq 0$, there is no budget for nodes in set $F$ and $D$, which means $F$ and $D$ are both empty. According to the greedy method, it only happens when $M=0$ which violates our assumption. Therefore, $B-\sum_{i\in E}\overline{m}_i-m_d>0$.
\end{proof}

\begin{lemma}\label{lemma:rhoF}
For any nonnegative $\rho_d$, there exists an optimal solution  for (\ref{snegame1})-(\ref{snegame2}) such that $\forall i\in F$, there is at most ONE $m_i<\overline{m}_i$ and all the other $m_i=\overline{m}_i$. 
\end{lemma}

The proof of Lemma~\ref{lemma:rhoF} is in Section~\ref{subsec:sgpropertyproof}. Lemmas~\ref{remark:rho} - \ref{lemma:rhoF} establish the foundation for the following key result about the optimal defense strategy of (\ref{snegame1})-(\ref{snegame2}).
\begin{proposition}\label{proposition:dp}
For any nonnegative $\rho_d$, there exists an optimal solution $\{m_i\}_{i=1}^n$ such that
\begin{enumerate}
\item $\forall i\in F$, there is at most one $m_i<\overline{m}_i$ and all the other $m_i=\overline{m}_i$;
\item $m_d=\overline{m}_d$
\item $\forall i\in E$, $m_i=\overline{m}_i$;
\item $\forall i\in G$, $m_i=0$.
\end{enumerate}
\end{proposition}
We denote the node whose $m_i<\overline{m}_i$ in the first property of Proposition~\ref{proposition:dp} as node $f$ and its defending frequency as $m_f$. Based on Proposition~\ref{proposition:dp}, we can easily compute the value of $m_i$ for each node (except $m_f$) after the set allocation is fixed.
Also, we can explicitly list the defender's payoff, defender's budget usage and attacker's budget usage by putting each node into different sets as shown in Table~\ref{tbl:setallocation}.
\begin{table}[!t]
\centering
\caption{Nodes in Different Sets with Given $\rho_d$}
\renewcommand{\arraystretch}{1.05}
\label{tbl:setallocation}
\small{
\begin{tabular}{|c|c|c|c|}
\hline \multicolumn{1}{|c|}{} & \multicolumn{1}{|c|}{$F$} & \multicolumn{1}{|c|}{$E$} & \multicolumn{1}{|c|}{$G$}\\
\hline
Defender's & \multirow{2}{*}{$\overline{m}_i(r_i\alpha_i-C_i^D)-r_i$} & \multirow{2}{*}{$-\overline{m}_i C_i^D$} & \multirow{2}{*}{$-r_i$}\\
payoff&&&\\
\hline
Defender's & \multirow{2}{*}{$\overline{m}_i$} & \multirow{2}{*}{$\overline{m}_i$} & \multirow{2}{*}{$0$}\\
budget usage&&&\\
\hline
Attacker's & \multirow{2}{*}{$\overline{m}_i \alpha_i$} & \multirow{2}{*}{$0$} & \multirow{2}{*}{$0$}\\
budget usage&&&\\
\hline
\end{tabular}
}
\end{table}
For the fractional node, its $m_i$ can be computed using linear programming when all the other $m_i$ have been determined. We use dynamic programming to determine the optimal set allocation.

From the discussion above, we propose the following algorithm to the defender's problem (see Algorithm~\ref{alg:snefinal}). The algorithm iterates over all possible node $d$ in set $D$ and all possible node $f$ with fractional assignment in set $F$. We first compute a special case when set $G$ is empty (line 2). In this case, the defender's optimal strategy can be obtained by solving (\ref{eqa:nog}) based on Proposition~\ref{proposition:dp}.
\begin{equation}\label{eqa:nog}
\begin{aligned}
Val(d,f)=&\max_{\rho_d,p_i,m_f}\sum_{i\neq f} \overline{m}_i(p_ir_i\alpha_i-C_i^D)\\
&-p_ir_i+m_f(r_f\alpha_f-C_f^D)-r_f\\
s.t.\  &\sum_{i\neq f}\overline{m}_i+m_f\leq B,\ \ \sum_{i\neq f}p_i\overline{m}_i\alpha_i+m_f\alpha_f\leq M\\
       \rho_d=&\frac{r_i-\overline{m}_i(r_i\alpha_i+C_i^A)}{\overline{m}_i\alpha_i}\geq 0,\ 0\leq p_i \leq 1,\ m_f\geq 0
\end{aligned}
\end{equation}
The algorithm then iterates over nonnegative $\rho_d$ with a step size $\rho_{step}$ (line 4). Given $\rho_d, d, f$, the best set allocation (together with $m_i$ for all $i$) are determined using dynamic programming as explained below. 

For any given $\rho_d$, $d$ and $f$ , we compute $\overline{m}_i$ for all $i$ (line 5). Let $SEQ(i,b,m,d,f,ind)$ denote the maximum payoff of the defender considering only node $1$ to node $i$ (excluding nodes $d$ and $f$), for a given defender's budget $b \in [0,B]$ and an attacker's budget $m \in [0,M]$ . The parameter $ind$ is a boolean variable that indicates whether we can put nodes in set $E$ arbitrarily. If $ind$ is $True$, any node (except nodes $d$ and $f$) can be in set $E$. Otherwise, a node $i$ can be allocated to set $E$ only if $r_i-\overline{m}_i(r_i\alpha_i+C_i^A)\leq 0$. 
The value of $SEQ(i,b,m,d,f,ind)$ is determined recursively. If node $i$ is either $d$ or $f$, we simply set $SEQ(i,b,m,d,f,ind) = SEQ(i-1,b,m,d,f,ind)$. Otherwise, we have the following recurrence equation, where the three cases refer to putting node $i$ in sets $F$, $E$ and $G$,  respectively.\ifreport based on Table~\ref{tbl:setallocation}:\fi
\begin{equation}
\begin{aligned}
&SEQ(i,b,m,d,f,ind)=\\
&\max\begin{cases}SEQ(i-1,b-\overline{m}_i,m-\alpha_i\overline{m}_i,d,f,ind)\\ \hspace{20ex} +\overline{m}_i(r_i\alpha_i-C_i^D)-r_i\\
                   SEQ(i-1,b-\overline{m}_i,m,d,f,ind)-\overline{m}_iC_i^D\\
                   SEQ(i-1,b,m,d,f,ind)-r_i
\end{cases}
\end{aligned}
\end{equation}
We have the following boundary conditions:
\begin{enumerate}
\item The recursion $SEQ$ will return $-\infty$ when $i>0$ and (i) $m<0$, or (ii) $b<0$, or (iii) $m=0$ and $ind=False$;

\item $SEQ(0,b,m,d,f,True)$ returns the solution to the following problem (i.e., the total payoffs contributed by nodes $d$ and $f$):
\begin{equation}
\label{residualbudget}
\begin{aligned}
\max_{m_f} &\  m_f(r_f\alpha_f-C_f^D)-r_f+\overline{m}_d(pr_d\alpha_d-C_d^D)-pr_d\\
s.t.\ \  & m_f+\overline{m}_{d}\leq b\\
& m-\overline{m}_d\alpha_d\leq m_f\alpha_f \leq m\\
& m_f\leq \overline{m}_f\\
& p=\frac{m-m_f\alpha_f}{\alpha_d\overline{m}_{d}}
\end{aligned}
\end{equation}
\item Similarly, $SEQ(0,b,m,d,f,False)$ returns the solution to the following problem:
\begin{equation}
\label{residualbudget2}
\begin{aligned}
\max_{m_f} &\ m_f(r_f\alpha_f-C_f^D)-r_f+\overline{m}_d(r_d\alpha_d-C_d^D)-r_d\\
s.t.\ \  & m_f+\overline{m}_{d}\leq b\\
& m_f\alpha_f \leq m-\alpha_d\overline{m}_{d}\\
& m_f\leq \overline{m}_f
\end{aligned}
\end{equation}
Note that if the constraints in (\ref{residualbudget}) or (\ref{residualbudget2}) define an empty set for $m_f$, $SEQ$ simply returns $-\infty$.


\end{enumerate}
\begin{algorithm}
\caption{Sequential Strategy for Defender}
\label{alg:snefinal}
\begin{algorithmic}[1]
\For{$d,f\leftarrow\ 1\ to\ n$}
\State Solve (\ref{eqa:nog}) to obtain $Val(d,f)$
\State $\rho_{max}\leftarrow \rho: \sum_{i=1}^n \alpha_im_i(\rho)= M$
\For{$\rho_d \leftarrow$ $0$ to $\rho_{max}$ with step size $\rho_{step}$}
\State $\overline{m}_i\leftarrow m_i(\rho_d)$ for all $i$
\State $val'_{d,f,\rho_d}\leftarrow SEQ(n,B,M,d,f,True)$
\State $val''_{d,f,\rho_d}\leftarrow SEQ(n,B,M,d,f,False)$
\EndFor
\State $P_{dp}(d,f)\leftarrow \max_{\rho}\{val'_{d,f,\rho_d},val''_{d,f,\rho_d}\}$
\EndFor
\State $P_{alg}\leftarrow \max_{d,f}\{P_{dp}(d,f),Val(d,f)\}$
\end{algorithmic}
\end{algorithm}

Algorithm~\ref{alg:snefinal} computes the optimal solution by searching over all combinations of $d$, $f$ and $\rho_d$. For any given combination, the dynamic program actually finds all the solutions that satisfy Proposition~\ref{proposition:dp}, meaning that $P_{dp}(d,f)$ returns the optimal defense strategy under given $d$, $f$ and $\rho_d$ (line 9). Therefore, $P_{alg}$ is the maximum payoff that the defender can achieve (line 11). For the dynamic program, we round the input before running $SEQ(n,B,M,d,f,ind)$, since the recursion may never stop without rounding. Denote $\delta$ as the rounding parameter, we have $\overline{m}_i\leftarrow \left\lfloor \frac{\overline{m}_i}{\delta} \right\rfloor$, $\alpha_i\leftarrow \left\lfloor \frac{\alpha_i}{\delta} \right\rfloor$ for all $i$ and $B \leftarrow \left\lfloor \frac{B}{\delta} \right\rfloor$, $M \leftarrow \left\lfloor \frac{M}{\delta} \right\rfloor$. By setting $\delta$ small enough, Algorithm~\ref{alg:snefinal} can find a strategy that is arbitrarily close to the subgame perfect equilibrium strategy of the defender. Formally, we can establish the following result.

\begin{restatable}{theorem}{primethmrounding}
\label{thm:rounding}
Let $|P_{alg}|$ denote the defender's cost obtained by Algorithm~\ref{alg:snefinal} and $|P^\star|$ the optimal cost. Given $\rho_{step}$ and the rounding parameter $\delta$, We have $\frac{|P_{alg}|}{|P^\star|}\leq 1+(\rho_{step}+\delta)O(N)$.
\end{restatable}
Please find the detailed proof in Section Theorem~\ref{thm:rounding} provides the performance guarantee of Algorithm~\ref{alg:snefinal} showing the trade-off between performance and the time complexity. Based on Theorem~\ref{thm:rounding}, we have the following corollary. 

\begin{corollary}\label{cor:complexity}
By setting both $\rho_{step}$ and $\delta$ with $O(\frac{1}{N})$, Algorithm~\ref{alg:snefinal} can achieve a near-optimal solution and its complexity is $O(N^5 BM)$
\end{corollary}
\ifreport
\begin{proof}
The complexity of the dynamic recursion $SEQ(\cdot)$ is $O(NBM)$ since $SEQ(\cdot)$ iterates all $B$, $M$ and $N$. Algorithm~\ref{alg:snefinal} iterates $d$ and $f$ with $\rho_{step}$ and rounding parameter $\delta$. Thus, its complexity is $O(\frac{N^3 BM}{\rho_{step}\cdot \delta})$. By setting both $\rho_{step}$ and $\delta$ with $O(\frac{1}{N})$, the total complexity of Algorithm~\ref{alg:snefinal} is $O(N^5 BM)$.
\end{proof}
\fi

\section{Numerical Results}\label{sec:numerical}
In this section, we present numerical results for our game models. For the illustrations, we assume that all the attacking times $\alpha_i$ are deterministic as in Sections~\ref{sec:sequential}.
We study the payoffs of both the attacker and the defender and their strategies in both Nash Equilibrium (two-node setting) and subgame perfect equilibrium (both two-node and five-node settings), and study the impact of various parameters including resource constraints $B$, $M$, and the unit value $r_i$.

\subsection{Simulations with Selected Parameters}\label{subsec:experiments1}
\begin{figure}
\captionsetup[subfigure]{justification=centering}
    \centering
      \begin{subfigure}{0.24\textwidth}
        \includegraphics[width=\textwidth]{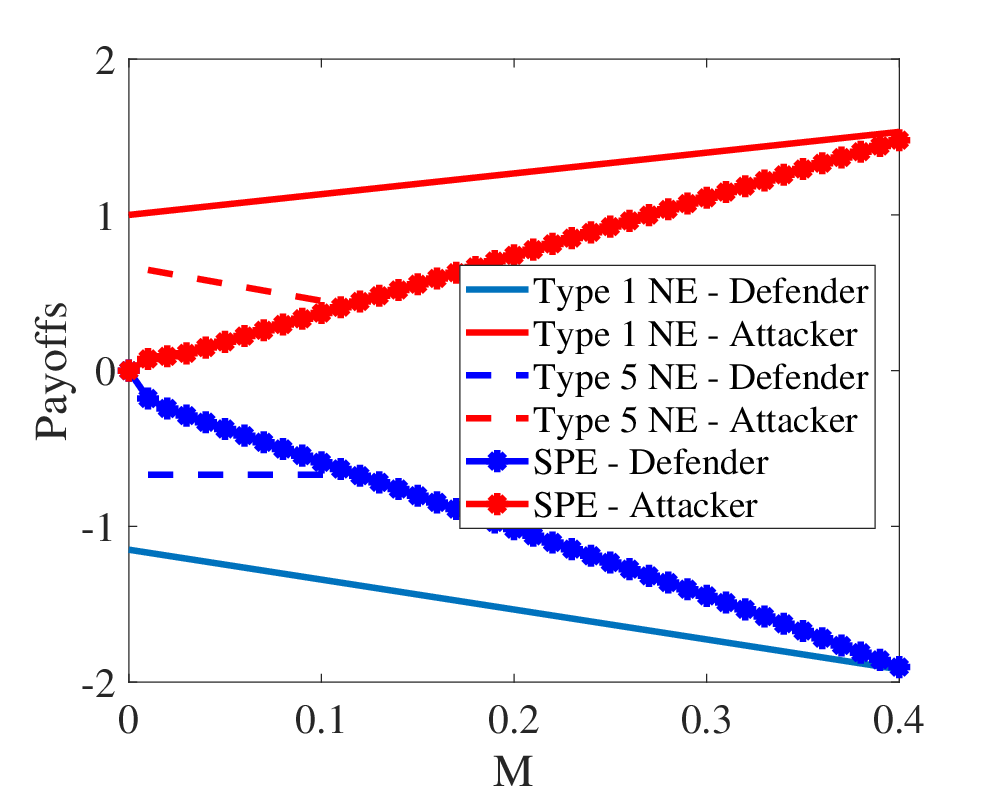}
          \caption{Payoffs with varying $M$}
          \label{fig:subfig1:a}
      \end{subfigure}
      \begin{subfigure}{0.24\textwidth}
        \includegraphics[width=\textwidth]{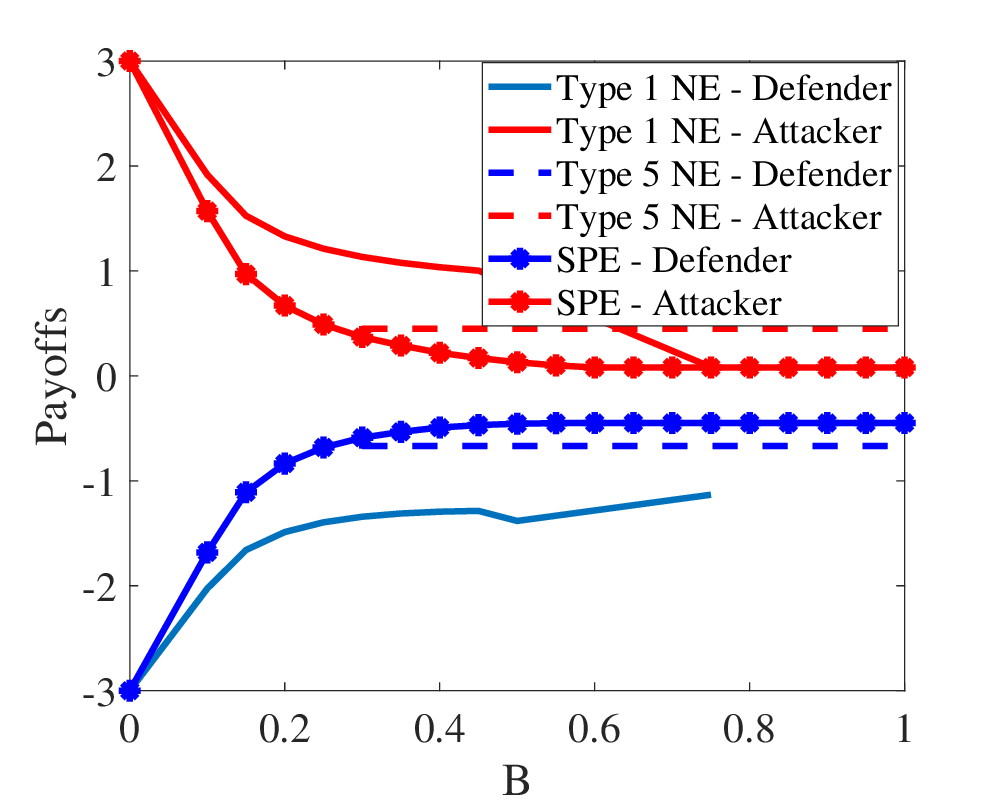}
          \caption{Payoffs with varying $B$}
          \label{fig:subfig1:b}
      \end{subfigure}
    \caption{The effects of varying resource constraints on payoffs. In both figures, $r_1=2, r_2=1, w_1=1.7, w_2=1.6, C_1^D=0.5, C_2^D=0.6, C_1^A=1, C_2^A=1.5$, $B=0.3$ in (a), and $M=0.1$ in (b)}
    \label{fig:subfig}
\end{figure}

We first study the impact of the resource constraints $M$ and $B$ on the player's payoffs in a two-node setting. The results are given in Figure~\ref{fig:subfig}, where we have plotted both Type 1 and Type 5 NEs~\footnote{There are also Type 2 NEs, which are omitted for the sake of clarify.} and subgame perfect equilibria. A Type 5 NE only occurs when $M$ is small as shown in Figure~\ref{fig:subfig1:a}, while Type 1 NE appears when $B$ is small as shown in Figure~\ref{fig:subfig1:b}, which is expected since $B$ is fully utilized in a Type 1 NE while $M$ is fully utilized in a Type 5 NE. When the defense budget $B$ becomes large, the summation of $m_i$ does not necessarily equal to $B$ and thus Type 1 NEs disappear. Similarly, Type 5 NEs disappear for large attack budget $M$. In both figures, the subgame perfect equilibria always bring the defender higher payoffs compared with Nash Equilibria, which is expected.

\subsection{Simulations with Real-World data} %
\begin{figure}
\captionsetup[subfigure]{justification=centering}
    \centering
      \begin{subfigure}{0.24\textwidth}
        \includegraphics[width=\textwidth]{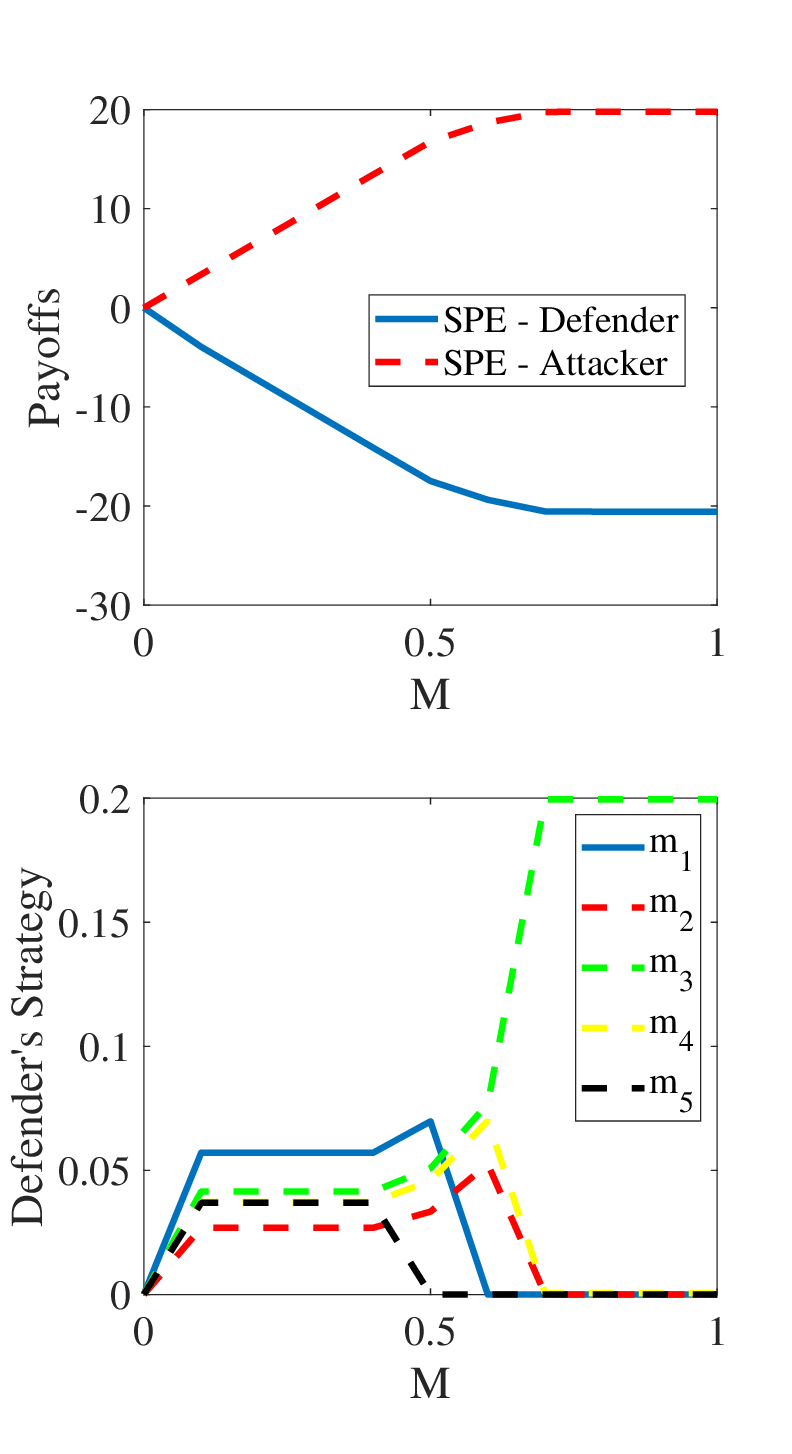}
          \caption{Payoffs and strategies with varying $M$}
          \label{fig:sgchangewithm2}
      \end{subfigure}
      \begin{subfigure}{0.24\textwidth}
        \includegraphics[width=\textwidth]{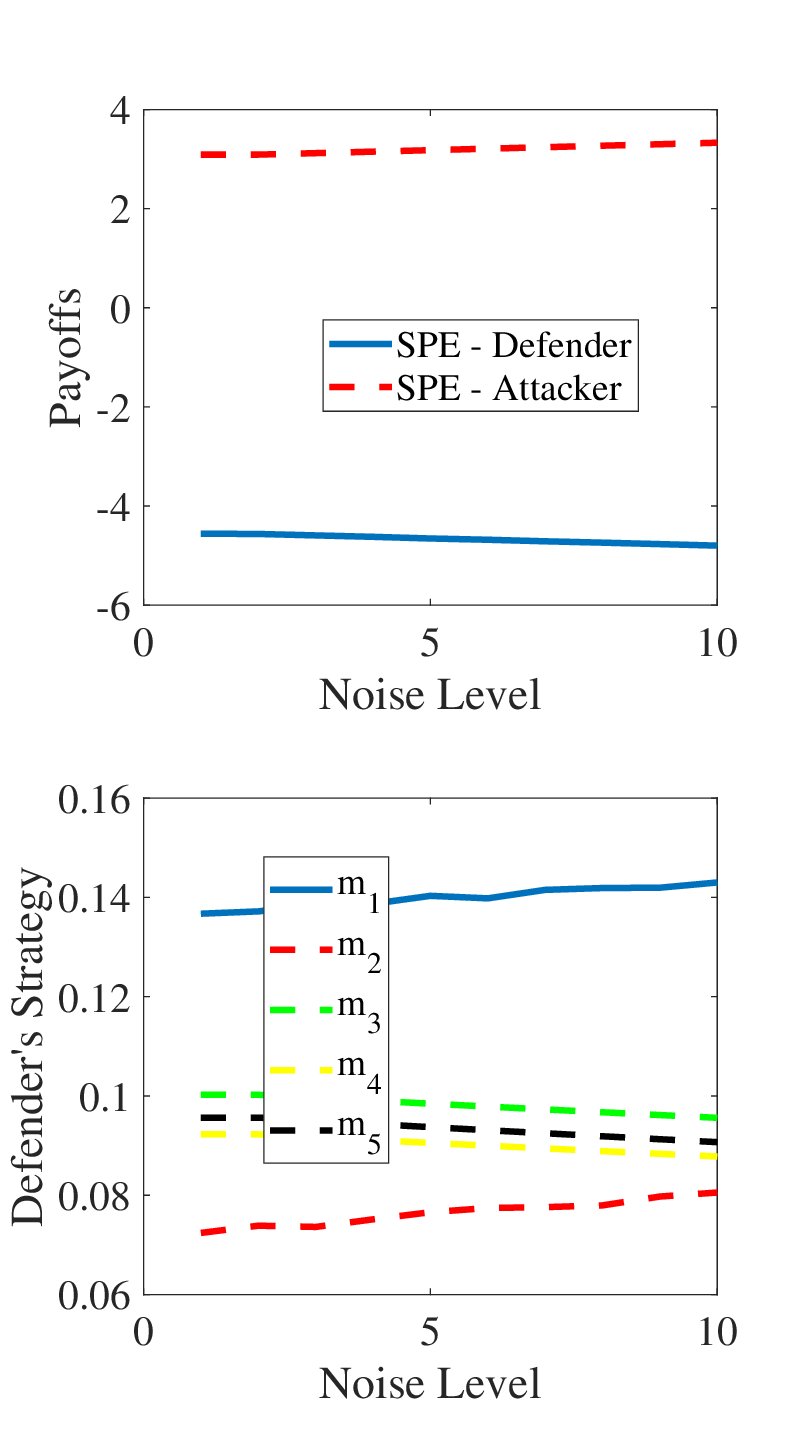}
          \caption{Payoffs and strategies with varying noise level}
          \label{fig:sgchangewithr2}
      \end{subfigure}
    \caption{The effects of varying resource constraint $M$ and unit value $r_i$, where $B=0.2$ in (a), $B=0.5$ and $M=0.3$ in (b). In (b), a random noise level is added to $r_1$ and $r_2$.}
\end{figure}

To have a better understanding of the performance of Algorithm~\ref{alg:snefinal}, we consider a five-node setting and use real-word data from the National Vulnerability Database (NVD)\cite{nvdurl}. We pick five vulnerability incidents about IoT devices revealed by the database. For each incident, we use their Impact Score (the potential impact of the vulnerability), Exploitability Score (how vulnerable the thing itself is to attack), Vulnerability Base Score (how critical the vulnerability is) and Attack Complexity (Low or High) \cite{CVSSuser1,CVSSuser2,CVSSuser3,CVSSuser4,CVSSuser5} as an approximation of the node value, attacking time, defending cost and attacking cost respectively. Specifically, we set node values as $r=[5.9\ 3.6\ 5.9\ 5.2\ 3.6]$. For the attacking times, since higher Exploitability Score means easier attack, we take the reciprocal and set $\alpha=[10/3.9\ 10/2.8\ 10/2.8\ 10/2.8\ 10/3.9]$ where the constant $10$ is used for normalization. The Vulnerability Base Score is utilized to approximate the defending cost by setting $C^D = [9.8\ 6.5\ 8.8\ 8.1\ 7.5]/3$, while the attacking cost is set to $2$ if the Attack Complexity is High and $1$ otherwise. We study the effects of varing $M$ and $r$ in Figure~\ref{fig:sgchangewithm2}.

In Figure~\ref{fig:sgchangewithm2}, the attacker's budget $M$ varies from $0$ to $1$ and the defending budget $B=0.2$. When $M=0$, the defender can set $m_i$ for all $i$ to arbitrary small (but positive) values, so that the attacker is unable to attack any node, leading to a zero payoff for both players. As $M$ becomes larger, the attacker's payoff increases, while the defender's payoff decreases, and the defender tends to defend the nodes with higher values more frequently, as shown in Figure~\ref{fig:sgchangewithm2}(lower). The defender gradually stop protecting low value nodes and move all the resources to defend node $3$. Note that the defending frequency for node $3$ is smaller than that for node 1 at the beginning. This is because when $M$ is small, the attacker attacks each node with a very small probability, thus the defender can protect all the nodes at the same time to prevent big loss. Since node $1$ and $3$ have the same unit value while $\alpha_1 < \alpha_3$, the defender protects node $1$ more frequently. However, when the attacker has enough resources to attack each node with a much higher probability, it is not beneficial for the defender to protect other nodes except node $3$ since it has the highest node value and attacking time.

In Figure~\ref{fig:sgchangewithr2}, we fixed $r_3$ through $r_5$ and increase $r_1$ and $r_2$ by adding a random noise uniformly distributed between $[\text{noise\_level}-1, \text{noise\_level}] * 0.1$. We vary the $\text{noise\_level}$ from $1$ to $10$. As shown in the figure, $m_1$ and $m_2$ keep increasing when the noise level becomes larger, while the defending frequencies for all other three nodes decrease due to limited defending resources, which indicates that the defender should protect the nodes with higher values more frequently in the subgame perfect equilibrium.

\begin{table}[!t]
\centering
\caption{Running Time Improvement}
\renewcommand{\arraystretch}{1.05}
\label{tbl:runtime}
\small{
\begin{tabular}{|c|c|c|}
\hline \multicolumn{1}{|c|}{No. of nodes} & \multicolumn{1}{|c|}{Algorithm~\ref{alg:snefinal}} & \multicolumn{1}{|c|}{Algorithm 1 in~\cite{zhang2015game}}\\
\hline
2 & 12.1 sec & 31.4 sec\\
\hline
3 & 131.5 sec & 410.8 sec\\
\hline
4 & 1036.3 sec & 3710.8 sec\\
\hline
5 & 2117.9 sec & 9261.3 sec\\
\hline
6 & 3.45 hours & 24 - 26 hours\\
\hline
\end{tabular}
}
\end{table}

Table~\ref{tbl:runtime} compares the running time of Algorithm~\ref{alg:snefinal} and that of the corresponding algorithm in our conference paper~\cite{zhang2015game}. All experiments are conducted on a desktop with 4-Core Intel i5-4670K CPU @ 3.40GHz and Matlab R2019a. The same simulation setting as in Figure~\ref{fig:sgchangewithm2} is applied with fixed $M=0.2$. We observe that Algorithm~\ref{alg:snefinal} is much faster than the original algorithm in our conference paper and the improvement is more significant in a larger setting. 

\ignore{
We make following observations from our numerical result:
\begin{enumerate}
\item The defender may prefer different types of NEs under different scenarios, so does the attacker.
\item Subgame perfect Equilibrium always bring more benefits to the defender compared to Nash Equilibrium.
\item The defender should protect the nodes with higher values more frequently in both Nash and subgame perfect Equilibrium.
\item In addition to the defender's constraint $B$, the attacker's resource constraint $M$ can also have a significant impact on the defender's behavior.
\item When $M$ is large, protecting high value nodes more frequently and giving up several low value nodes is more beneficial for the defender compared to defending all the nodes with low frequency.
\end{enumerate} }

\section{Conclusion}\label{sec:conclusion}
In this paper, we propose a two-player non-zero-sum game for protecting a system of multiple components against a stealthy attacker where the defender's behavior is fully observable and both players have strict resource constraints. We prove that periodic defense and non-adaptive $i.i.d.$ attack are a pair of best-response strategies with respect to each other in the space of both adaptive and non-adaptive strategies. For this pair of strategies, we characterize the set of Nash Equilibria of the game, and show that there is always one (and maybe more) equilibrium, for the case when the attack times are deterministic.
We further study the sequential game where the defender first publicly announces its strategy and design an algorithm that can identify a strategy that is arbitrarily close to the subgame perfect equilibrium strategy for the defender. We also provide a full analysis of the algorithm performance and its complexity guarantee.

\bibliographystyle{abbrv}
\bibliography{ref}
\section{Appendix}\label{sec:appendix}
\subsection{Proof of Lemma~\ref{lemma:attackerlemma2}}\label{subsec::proofofattacklemma}
\begin{proof}
In order to get the attacker's best responses against any defender's deterministic strategies, we can divide \eqref{generalattacker4} into $N*L$ sub-optimization problems
\begin{equation}
\begin{aligned}
\label{generalattacker7}
\min_{W_{i,k}}& \frac{E[\min(W_{i,k}+\alpha_{i,k},X_{i,k})] r_i+P(W_{i,k}<X_{i,k}) C_i^A}{T}\\
s.t.&\frac{E[\min(W_{i,k}+\alpha_{i,k},X_{i,k})]-E[\min(W_{i,k},X_{i,k})]}{T}\leq M_{i,k}
\end{aligned}
\end{equation}
where $\sum_{i=1}^{N}\sum_{k=1}^{L_i} M_{i,k}=M$ and $M_{i,k}$ can be arbitrary positive number. Note that we consider the equivalent minimization problem by taking the negative of the target function of (\ref{generalattacker1}) and omitting the constant part. We claim that, the optimal solution to (\ref{generalattacker7}) is to allocate as much budget as possible to $P(W_{i,k}=0)$, that is
\begin{equation}
\label{solution1}
W_{i,k}^*=\begin{cases}
0 &\text{    w.p.}\ p_{i,k}^*\\
\geq X_{i,k}  &\text{    w.p.}\ 1-p_{i,k}^*
\end{cases}
\end{equation}
where $p_{i,k}^*=\min(1,\frac{M_{i,k}T}{E[\min(\alpha_{i,k},X_{i,k})]})$ if $r_i (E[\min(\alpha_{i,k},X_{i,k})]-X_{i,k})+C_i^A < 0$, and $p_{i,k}^*=0$ otherwise.

Since $M_{i,k}$ is any number such that $\sum_{i=1}^{N}\sum_{k=1}^{L_i} M_{i,k}=M$, the optimal solution of \eqref{generalattacker4} also satisfies the same structure of \eqref{solution1}. We then prove our claim. For simplicity, we assume that $W_{i,k}$ is a discrete r.v., and without loss of generality, it has the following p.m.f
\begin{equation}\label{wequation}
W_{i,k}=\begin{cases}
0 & \text{w.p. $p_0$}\\
v_i & \text{w.p. $p_i,\ i=1\cdots n$}\\
\geq X_{i,k} & \text{w.p. $1-\sum_{j=0}^n{p_j}$}\\
\end{cases}
\end{equation}
where $n\in \mathbb{N}$ such that $0<v_1<v_2<\ldots <v_n<X_{i,k}$. The following proof can be adapted to the continuous $W_{i,k}$ as well by replacing sums with integrals and p.m.f with p.d.f.

Putting \eqref{wequation} into \eqref{generalattacker7},
attacker's problem can then be converted to the following form
{\allowdisplaybreaks
\begin{equation}
\label{lemma3 5}
\begin{aligned}
\min \sum_{j=0}^n p_j(r_i [E[\min(v_j+\alpha_{i,k},X_{i,k})]-X_{i,k}]+C_i^A)+X_{i,k} r_i
\end{aligned}
\end{equation}}
with two constraints: $\sum_{j=0}^n p_j E[\min(\alpha_{i,k},X_{i,k}-v_i)] \leq M_{i,k}T$ and $\sum^n_{j=0} p_j \leq 1$. where $v_0=0$.

Let $J(\{p_0,...,p_n\})$ denote the objective function in~\eqref{lemma3 5}. Since $r_i (E[\min(\alpha_{i,k},X_{i,k})]-X_{i,k})+C_i^A < r_i (E[\min(v_j+\alpha_{i,k},X_{i,k})]-X_{i,k})+C_i^A$, if $r_i (E[\min(\alpha_{i,k},X_{i,k})]-X_{i,k})+C_i^A \geq 0$, $J(\{p_0,...,p_n\})$ is minimized by setting $p_j = 0, \forall j=0,...,n$, which implies $W_{i,k}\geq X_{i,k}$ w.p.1. Such condition describes the case that even if the attacker attacks the node immediately after it is recovered, its reward is still less than 0. Therefore, the attacker never attacks. If $r_i (E[\min(\alpha_{i,k},X_{i,k})]-X_{i,k})+C_i^A <0$, we claim that the optimal solution is to allocate as much budget $M_{i,k}T$ as possible to $p_0$, that is, we set all $p_j=0$, $1\leq j\leq n$, and $p_0=\min(1,\frac{M_{i,k}T}{E[\min(\alpha_{i,k},X_{i,k})]})$. This is clearly true if $r_i (E[\min(v_j+\alpha_{i,k},X_{i,k})]-X_{i,k})+C_i^A \geq 0$. Therefore, it suffices to consider the case when $r_i (E[\min(\alpha_{i,k},X_{i,k})]-X_{i,k})+C_i^A < r_i (E[\min(v_j+\alpha_{i,k},X_{i,k})]-X_{i,k})+C_i^A < 0$.

To prove the claim, consider an optimal solution $\{p_0, p_1, ..., p_n\}$ to~\eqref{lemma3 5}. We show that if $p_0 < \min(1,\frac{M_{i,k}T}{E[\min(\alpha_{i,k},X_{i,k})]})$, then we can find another optimal solution $\{p'_0, p'_1, ..., p'_n\}$ such that $p'_0 > p_0$. We distinguish the following two cases:\\

\noindent Case 1: $p_0 E[\min(\alpha_{i,k},X_{i,k})] + \sum_{j=1}^n p_j E[\min(\alpha_{i,k},X_{i,k}-v_i)] < M_{i,k}T$. Then by the optimality of $\{p_0, p_1, ..., p_n\}$ and the assumption that $r_i (E[\min(v_j+\alpha_{i,k},X_{i,k})]-X_{i,k})+C_i^A<0$, we must have $\sum^n_{j=0} p_j = 1$. Let $j \geq 1$ denote an index such that $p_j>0$. Then there must exist a small amount $\triangle p > 0$ such that $p'_0 = p_0 + \triangle p, p'_j = p'_j - \triangle p, p'_k = p_k, \forall k \neq 0$ and $k \neq j$ is again a feasible solution to~\eqref{lemma3 5}. We further have
{\allowdisplaybreaks
\begin{align*}
&J(\{p_0,...,p_n\}) - J(\{p'_0,...,p'_n\})\\
&=\triangle p(r_i[E[\min(v_j+\alpha_{i,k},X_{i,k})]-X_{i,k}]+C_i^A)\\
&\ \ \ -\triangle p(r_i[E[\min(\alpha_i,X_{i,k})]-X_{i,k}]+C_i^A)\\
&=\triangle p r_i(E[\min(v_j+\alpha_{i,k},X_{i,k})]-E[\min(\alpha_{i,k},X_{i,k})])\\
&\geq 0
\end{align*}}
\noindent Case 2: $p_0 E[\min(\alpha_{i,k},X_{i,k})] + \sum_{j=1}^n p_j E[\min(\alpha_{i,k},X_{i,k}-v_i)] = M_{i,k}T$. Again let $j\geq 1$ denote an index such that $p_j>0$. Then there must exist a small amount $\triangle M > 0$ such that $p'_0 = p_0 +\frac{\triangle M}{E[\min(\alpha_{i,k},X_{i,k})]}, p'_j = p_j - \frac{\triangle M}{E[\min(\alpha_{i,k},X_{i,k}-v_j)]}, p'_k = p_k, \forall k \neq 0$ and $k \neq j$ is a feasible solution to~\eqref{lemma3 5}. We further have
{\allowdisplaybreaks
\begin{align*}
&J(\{p_0,...,p_n\}) - J(\{p'_0,...,p'_n\})\\
&=\frac{\triangle M (r_i[E[\min(v_j+\alpha_{i,k},X_{i,k})]-X_{i,k}]+C_i^A)}{E[\min(\alpha_{i,k},X_{i,k}-v_j)]}\\
&\ \ \ -\frac{\triangle M (r_i[E[\min(\alpha_{i,k},X_{i,k})]-X_{i,k}]+C_i^A)}{E[\min(\alpha_{i,k},X_{i,k})]}\\
&=\frac{\triangle M}{E[\min(\alpha_{i,k},X_{i,k}-v_j)]}(r_iv_j-r_iX_{i,k}+C_i^A)\\
&\ \ \ -\frac{\triangle M}{E[\min(\alpha_{i,k},X_{i,k})]}(-r_iX_{i,k}+C_i^A)\\
&\geq 0
\end{align*}}
\end{proof}
\vspace{-4ex}
\subsection{Proof of Theorem~\ref{thm:markovattacker}}\label{subsec:markovattakcerproof}
\begin{proof}
When the attacker's strategy is an ergodic Markov chain, the $p_{i,k}$'s time-average distribution is the same as its steady state distribution. Therefore the defender's problem (\ref{generaldefender1}) can be transferred to the following
\begin{equation}
\begin{aligned}
\label{markovdefender2}
&\max_{\{X_{i,k}\},L_i}\lim_{T\rightarrow \infty}E\bigg[ \sum_{i=1}^N\bigg(-\frac{L_i C_i^D+Tr_i}{T}\\
                      &+\frac{(\sum_{k=1}^{L_i} E[p_{i,k}]\min(\alpha_{i,k},X_{i,k})+(1-E[p_{i,k}])X_{i,k})\cdot r_i}{T}\bigg) \bigg] \\
\end{aligned}
\end{equation}
with the same resource constraint in \eqref{eqa:defenderconstraint} where the expectation in the numerator is with respect to the steady-state distribution of $p_{i,k}$. We find that \eqref{markovdefender2} is the same as \eqref{generaldefender1} if we set
\begin{equation}
W_{i,k}^\star=\begin{cases}
0\ \ \ &w.p.\ E[p_{i,k}]\\
\infty\ \ \ &w.p.\ 1-E[p_{i,k}]
\end{cases}
\end{equation}
Here, $E[p_{i,k}]$ is the expected value of $p_{i,k}$'s steady state distribution. Therefore, based on Lemma~\ref{lemma:general} and Theorem~\ref{theorem:defender1}, we know that the periodic strategy is defender's best response.
\end{proof}
\vspace{-2ex}
\subsection{Proof of Theorem~\ref{thm:defendermarkov}}\label{subsec:markovproof}
\begin{proof}
For simplicity, we assume there is only one node and the attacking time $\alpha_{i,k}\ \forall i,k$ is deterministic. (We omit all the subscript $i$ in this proof since there is only one node and use $\alpha$ to represent $\alpha_{i,k}\ \forall k$). The defender's Markovian strategy has two states $x_1$ and $x_2$ referring the two defending periods whose transition probabilities are as follows: $P(X_{k+1}=x_2 | X_{k}=x_1)=u$ and $P(X_{k+1}=x_1 | X_{k}=x_2)=v$. Let $\pi_1$ and $\pi_2$ represent the probability that $X_k=x_1$ and $X_k=x_2$ in steady state, respectively. We have $\pi_1=\frac{v}{u+v}$ and $\pi_2=\frac{u}{u+v}$. Since the attacker can observe the defender' defending period, the attacking strategy may depend on the defender's state (the previous defending period). 
Let $p_1$ denote the attacking probability when the attacker observes the defender using $X_1$ in the previous defense move, and $p_2$ as the attacking probability for $X_2$. 

We compute the average payoff for the attacker per defense move. Given the defender uses $x_1$ in the previous defense move, the expected payoff for the attacker is $S_{X_{k-1}=x_1}=[(1-u)\cdot(x_1-\alpha)p_1+u\cdot(x_2-\alpha)p_1]\cdot r-p_1C^A$. If the defender uses $x_2$ in the previous defense move, the attacker's expected payoff is $S_{X_{k-1}=x_2}=[v p_2(x_1-\alpha)+(1-v)p_2(x_2-\alpha)]\cdot r+p_2C^A $. Here, we assume $x_1\geq \alpha$ and $x_2\geq \alpha$. (The defender has no incentive to set $x_1$ or $x_2$ smaller than $\alpha$). Further, since the Markov chain is time reversible, we also have $P(X_{k-1}=x_2 | X_{k}=x_1)=u$ and $P(X_{k-1}=x_1 | X_{k}=x_2)=v$. For attacker's budget constraint, we have
\begin{displaymath}
\begin{aligned}
&\frac{E[\sum_{k=1}^{L_i}\min(W_k+\alpha,X_k)-\min(W_k,X_k)]}{T}\\
=&\pi_1\cdot E[\min(W_k+\alpha,x_1)-\min(W_k,x_1)|X_k=x_1]\\
+&\pi_2\cdot E[\min(W_k+\alpha,x_2)-\min(W_k,x_2)|X_k=x_2]\\
=&\pi_1[(1-u)p_1\alpha+u p_2\alpha]+\pi_2[(1-v)p_2\alpha+v p_1\alpha]\\
=&(\pi_1 p_1+\pi_2 p_2)\alpha
\end{aligned}
\end{displaymath}
Then, the attacker's optimization problem becomes
\begin{equation}\label{eqa:markovattacker}
\begin{aligned}
\max_{p_1,p_2}&\ \pi_1S_{X_{k-1}=x_1}+\pi_2S_{X_{k-1}=x_2}\\
&s.t.\ \pi_1 p_1+\pi_2 p_2\leq M/\alpha\\
& 0\leq p_1,p_2\leq 1
\end{aligned}
\end{equation}
Since \eqref{eqa:markovattacker} is a fractional knapsack problem, it's easy to show that setting $p_1=p_2$ is not optimal in general,  meaning that the {\it i.i.d.} strategy is NOT the attacker's optimal response against Markovian defending strategy.

\end{proof}
\subsection{Proof of Lemma~\ref{lemma:rhoF}}\label{subsec:sgpropertyproof}
\begin{proof}
Suppose the set allocation and $\rho_d$ are fixed, which means $m_d$ and $\overline{m}_i\ \forall i$ are also fixed. According to Lemma~\ref{lemma:rho}, we can now convert (\ref{snegame1})-(\ref{snegame2}) to the following problem:
\begin{equation}
\label{snegame4}
\begin{aligned}
\max_{m_i,i\in F}&\sum_{i\in F} [m_i(r_i\alpha_i-C_i^D)-r_i]- \sum_{i\in G}r_i-\sum_{i\in E}\overline{m}_iC_i^D\\
            &+m_d(pr_d\alpha_d-C_i^D)-pr_d\\
\end{aligned}
\end{equation}
with constraints: $\sum_{i\in F}m_i\leq B-\sum_{i\in E}\overline{m}_i-m_d$, $\sum_{i\in F}\alpha_im_i+p\alpha_dm_d\leq M$ and $0\leq m_i\leq \overline{m}_i\ \forall i\in F$.
where $p=\min\{1,\frac{M-\sum_{i\in F}\alpha_im_i}{\alpha_dm_d}\}$.\\

Case 1: If $\frac{M-\sum_{i\in F}\alpha_im_i}{\alpha_dm_d}\leq 1$, we put $p=\frac{M-\sum_{i\in F}\alpha_im_i}{\alpha_dm_d}$ back into the target function of (\ref{snegame4}) and convert it to
\begin{equation}
\label{snegame5}
\begin{aligned}
\max_{m_i,i\in F}&\sum_{i\in F} [m_i(r_i\alpha_i-C_i^D)-r_i]- \sum_{i\in G}r_i-\sum_{i\in E}\overline{m}_iC_i^D\\
            &+\frac{M-\sum_{i\in F}\alpha_im_i}{\alpha_dm_d}r_d(\alpha_dm_d-1)-m_d C_d^D\\
\end{aligned}
\end{equation}
with constraints: $\sum_{i\in F}m_i\leq B-\sum_{i\in E}\overline{m}_i-m_d$ and $0\leq m_i\leq \overline{m}_i\ \forall i\in F$.

It is easy to see that (\ref{snegame5}) is a fractional knapsack problem. Thus, there is at most one fractional variable which means at most one $m_i<\overline{m}_i$.\\

Case 2: If $\frac{M-\sum_{i\in F}\alpha_im_i}{\alpha_dm_d}>1$, the attacker's budget is not fully utilized and all $p^{\star}_i$ in (\ref{snegame1}) equal to $1$. Thus, the sets $D$ and $E$ are empty. Now suppose there exist two nodes $j$ and $k$ in $F$ with $m_j<\overline{m}_j$ and $m_k<\overline{m}_k$. Without loss of generality, by assuming $r_j\alpha_j-C_j^D\geq r_k\alpha_k-C_k^D$, we can always increase the defender's payoff by decreasing $m_k$ and increasing $m_j$ until either $m_j=\overline{m}_j$ or $m_k=0$. If $m_k=0$, node $k$ is in set $G$. Here, if the attacker's budget is fully utilized (as in Case 1), we can not guarantee the new payoff by decreasing $m_k$ and increasing $m_j$ is always bigger, since $\alpha_k$ may be much smaller than $\alpha_j$, making the increase of $m_j$ is very small due to limited attacker's budget.

Above all, we can claim that there exists an optimal solution with at most one node in set $F$ with $m_i<\overline{m}_i$.
\end{proof}

\vspace{-2ex}
\subsection{Proof of Theorem~\ref{thm:rounding}}\label{subsec:sgproof}
\begin{proof}
If the set $G$ is empty in the optimal solution $P^\star$, Algorithm~\ref{alg:snefinal} computes the optimal payoffs for the defender by solving (\ref{eqa:nog}). Then, we have $\max_{d,f}Val(d,f)=P^\star$. Therefore, $\frac{|P_{alg}|}{|P^\star|}=1$.

If the set $G$ is not empty in the optimal solution $P^\star$, we first consider the loss of performance due to $\rho_{step}$. Denote $\rho^\star$ as the optimal $\rho_d$ for computing $P^\star$ and $\rho'$ the first $\rho_d$ that is greater than $\rho^\star$ in Algorithm~\ref{alg:snefinal}. Let $\overline{m}^\star_i=m_i(\rho^\star)$ and $\overline{m}'_i=m_i(\rho')$. Let $|P_{\rho'}|$ refer to the total cost when $\rho^\star$ increases to $\rho'$ for the optimal solution $P^\star$. By increasing $\rho^\star$ to $\rho'$, each $\overline{m}^\star_i$ decreases to $\overline{m}'_i$ and the total cost increases in two parts. The first part is due to the decrease of $\overline{m}^\star_i$ for all $i$ in $F$. The second part comes from sets $E$ and $D$. Since $\overline{m}^\star_i$ decreases, the attacker has extra budget to attack the nodes in sets $E$ and $D$, moving these nodes to sets $F$ and $D$. For all sets $F$, $D$, $E$ and $G$ above, we refer to the set allocation in optimal solution $P^\star$. Let $H_F$ and $H_E$ denote the increase of total cost from the two parts, respectively. We have
\begin{align*}
H_F&=\sum_{i\in F}\bigg[r_i(1-\overline{m}'_i\alpha_i)+\overline{m}'_iC_i^D-r_i(1-\overline{m}^\star_i\alpha_i)-\overline{m}^\star_i C_i^D\bigg]\\
   &= \sum_{i\in F}\triangle \overline{m}_i(r_i\alpha_i-C_i^D)
\end{align*}
where $\triangle \overline{m}_i=\overline{m}^\star_i-\overline{m}'_i$.

Let $p'_d$ and $p^\star_d$ be the attacker's attacking probability for the node in set $D$ under $\rho'$ and $\rho^\star$, respectively. Denote $p'_i$ as the attacking probability for node $i$ under $\rho'$. We have
\begin{equation}\label{eqa:pstepproof1}
\begin{aligned}
H_E&=\sum_{i\in E}\bigg[p'_ir_i(1-\overline{m}'_i\alpha_i)+\overline{m}'_iC_i^D-\overline{m}^\star_i C_i^D\bigg]\\
   &+p'_dr_d(1-\overline{m}'_d\alpha_d)+\overline{m}_d'C_d^D-p^\star_dr_d(1-\overline{m}^\star_d\alpha_d)-\overline{m}^\star_d C_d^D\\
   &\leq \sum_{i\in E\cup D}p'_ir_i(1-\overline{m}'_i\alpha_i)
\end{aligned}
\end{equation}
Also note that $p'_i, i\in E \cup D$ must satisfy the resource constraint such that
\begin{equation}\label{eqa:pstepproof2}
\begin{aligned}
\sum_{i\in E\cup D} p'_i\overline{m}'_i\alpha_i\leq \sum_{i=1}^N \triangle \overline{m}_i\alpha_i
\end{aligned}
\end{equation}
where the right-hand side represents an upper bound on the extra budget for nodes in sets $E$ and $D$. From \eqref{eqa:pstepproof1} and \eqref{eqa:pstepproof2}, we have
\begin{displaymath}
H_E\leq \sum_{i=1}^N\triangle \overline{m}_i \alpha_i \cdot \max_i\{\frac{r_i(1-\alpha_i\overline{m}'_i)}{\alpha_i\overline{m}'_i}\}
\end{displaymath}

We further have
\begin{equation}
\begin{aligned}
\triangle \overline{m}_i&=\frac{r_i}{(\rho^\star+r_i)\alpha_i+C_i^A}-\frac{r_i}{(\rho'+r_i)\alpha_i+C_i^A}\\
             &\leq \frac{\rho_{step} r_i\alpha_i}{[(\rho^\star+r_i)\alpha_i+C_i^A][(\rho'+r_i)\alpha_i+C_i^A]}
\end{aligned}
\end{equation}
Since $\rho'$ is one of the $\rho_d$ that Algorithm~\ref{alg:snefinal} iterates through, we have $|P_{alg}|\leq |P_{\rho'}|$. Then, we can compute the approximation ratio as follows:

{\allowdisplaybreaks
\begin{align*}
&\frac{|P_{alg}|-|P^\star|}{|P^\star|} \leq \frac{|P_{\rho'}|-|P^\star|}{|P^\star|}= \frac{H_F+H_E}{|P^\star|}\\
&\leq\frac{\sum_{i\in F \cup D}\triangle \overline{m}_i(r_i \alpha_i-C_i^D)}{|P^\star|}\\
&+\frac{(\sum_{i=1}^N\triangle \overline{m}_i \alpha_i)\cdot \max_i\{\frac{r_i(1-\alpha_i\overline{m}'_i)}{\alpha_i\overline{m}'_i}\}}{|P^\star|}\\
& \leq \frac{\sum_{i\in F \cup D}\rho_{step}\frac{r_i\alpha_i(r_i\alpha_i-C_i^D)}{[(\rho^\star+r_i)\alpha_i+C_i^A]^2}}{|P^\star|}\\
&+\frac{(\sum_{i=1}^N \frac{\rho_{step} r_i\alpha_i^2}{[(\rho^\star+r_i)\alpha_i+C_i^A][(\rho'+r_i)\alpha_i+C_i^A]})\cdot \max_i\{\rho'+\frac{C_i^A}{\alpha_i}\}}{|P^\star|}\\
&\leq \rho_{step}\cdot \bigg( \frac{\sum_{i\in F\cup D}{\overline{m}^\star_i}^2(r_i\alpha_i-C_i^D)\alpha_i/r_i}{\sum_{i\in F\cup D}[r_i(1-\overline{m}^\star_i\alpha_i)+\overline{m}^\star_iC_i^D]}\\
&+ \frac{\sum_{i=1}^N\frac{r_i\alpha_i^2}{C_i^A\cdot C_i^A}\cdot \max_i\frac{C_i^A}{\alpha_i}+\sum_{i=1}^N\frac{r_i\alpha_i}{(\rho^\star+r_i)\alpha_i+C_i^A}}{\sum_{i\in G}r_i}     \bigg)\\
& \leq \rho_{step}\bigg(\max_i\{\frac{\overline{m}^\star_i\alpha_i(r_i\alpha_i-C_i^D)}{r_iC_i^D}\}\\
&+\frac{N\max_i\frac{r_i\alpha_i}{C_i^A}\max_i\frac{\alpha_i}{C_i^A}\max_i\frac{C_i^A}{\alpha_i}+N\max_i\frac{r_i\alpha_i}{C_i^A} }{\min_ir_i}\bigg)\\
&\leq \rho_{step}\left( \max_i\{\frac{\alpha_i}{C_i^D}\}+N\frac{\max_i\frac{r_i\alpha_i}{C_i^A}(1+\max_i\frac{\alpha_i}{C_i^A}\max_i\frac{C_i^A}{\alpha_i})}{\min_i r_i}  \right)\\
&\leq \rho_{step}\cdot O(N) \numberthis \label{rhostepeqa}
\end{align*}}
A similar argument can be used to bound the loss of performance due to rounding parameter $\delta$. The only difference is the decrease of $\overline{m}^\star_i$ which satisfies $\triangle \overline{m}_i\leq \frac{\rho_{step} r_i\alpha_i}{[(\rho^\star+r_i)\alpha_i+C_i^A]^2}+\delta$. The rest is very similar to (\ref{rhostepeqa}). It follows that $\frac{|P_{alg}|}{|P^\star|}\leq 1+(\rho_{step}+\delta)O(N)$ as desired.
\end{proof}

\end{document}